\PassOptionsToPackage{english}{babel}

\documentclass[a4paper,USenglish,numberwithinsect]{arxivlipics-v2018}
\nolinenumbers
\hideLIPIcs
\usepackage{microtype}%if unwanted, comment out or use option "draft"

\title{Best-case and Worst-case Sparsifiability of Boolean CSPs\footnote{An extended abstract of this work was accepted under the same title to the 13th International Symposium on Parameterized and Exact Computation (IPEC 2018).}}

\author{Hubie Chen}{Birkbeck, University of London}{hubie@dcs.bbk.ac.uk}{}{}
\author{Bart M.P. Jansen}{Eindhoven University of Technology}{b.m.p.jansen@tue.nl}{http://orcid.org/0000-0001-8204-1268}{Supported by NWO Gravitation grant ``Networks''.}
\author{Astrid Pieterse}{Eindhoven University of Technology}{a.pieterse@tue.nl}{http://orcid.org/0000-0003-3721-6721}{Supported by NWO Gravitation grant ``Networks''.}

\authorrunning{H. Chen, B.M.P. Jansen, and A. Pieterse}

\Copyright{Hubie Chen, Bart M.P. Jansen, and Astrid Pieterse}

\subjclass{%
,
,
}

%mandatory

\category{}%optional, e.g. invited paper

\relatedversion{}%optional, e.g. full version hosted on arXiv, HAL, or other respository/website

\supplement{}%optional, e.g. related research data, source code, ... hosted on a repository like zenodo, figshare, GitHub, ...

\funding{}%optional, to capture a funding statement, which applies to all authors. Please enter author specific funding statements as fifth argument of the \author macro.

\acknowledgements{We would like to thank Emil Je\v{r}\'abek for the proof of Lemma~\ref{lem:prime-powers-to-int}.}

\EventEditors{John Q. Open and Joan R. Access}
\EventNoEds{2}
\EventLongTitle{42nd Conference on Very Important Topics (CVIT 2016)}
\EventShortTitle{IPEC 2018}
\EventAcronym{IPEC}
\EventYear{2018}
\EventDate{December 24--27, 2016}
\EventLocation{Little Whinging, United Kingdom}
\EventLogo{}
\SeriesVolume{42}
\ArticleNo{23}

\usepackage{xspace}
\theoremstyle{plain}% format propositions/claims in the same was as theorems/lemmas (for some reason?)
\newtheorem{claim}[theorem]{Claim}
\newtheorem{observation}[theorem]{Observation}
\newtheorem{proposition}[theorem]{Proposition}
\newtheorem{assumption}[theorem]{Assumption}
\newcommand{\todo}[1][]{%
  \ifx/#1/%
    \textcolor{red}{TODO (astrid)!}%
  \else%
    \textcolor{red}{todo(astrid): #1}%
  \fi%
}

\let\plainqed\qedsymbol
\newcommand{\claimqed}{$\lrcorner$}
\newenvironment{claimproof}[1][\proofname]{\begin{proof}[#1]\renewcommand{\qedsymbol}{\claimqed}}{\end{proof}\renewcommand{\qedsymbol}{\plainqed}}

\newcommand{\weight}{\mathrm{weight}}

\newcommand{\yes}{\textsc{yes}\xspace}

\newcommand{\notcontainment}{\ensuremath{\mathsf{NP \not\subseteq coNP/poly}}\xspace}

\newcommand{\containment}{\ensuremath{\mathsf{NP  \subseteq coNP/poly}}\xspace}

\newcommand{\OR}{\textsc{or}\xspace}
\newcommand{\Oh}{\mathcal{O}}

\newcommand{\vect}[1]{\ensuremath{\mathbf{#1}}\xspace}
\newcommand{\rootCSP}{\textsc{$d$-Polynomial root CSP}\xspace}
\newcommand{\restr}{\upharpoonright}
\newcommand{\C}{\ensuremath{\mathcal{C}}\xspace}

\newcommand{\minor}{\mathsf{minor}}
\newcommand{\major}{\mathsf{major}}
\newcommand{\kor}[1]{\ensuremath{#1\mathrm{\textsc{-or}}}\xspace}

\newcommand{\csp}{\mathsf{CSP}}

\newcommand{\defparproblem}[4]{\par
 \vspace{3mm}
\noindent\fbox{
 \begin{minipage}{0.96\textwidth}
 \begin{tabular*}{\textwidth}{@{\extracolsep{\fill}}lr} #1 & {\bf{Parameter:}} #3 \vspace{1mm} \\ \end{tabular*}
 {\textbf{Input:}} #2
	\vspace{1mm}\\%
 {\textbf{Question:}} #4
 \end{minipage}
 }
 \vspace{3mm}
\par
}

\begin{document}

\maketitle

\begin{abstract}
We continue the investigation of polynomial-time sparsification for NP-complete Boolean Constraint Satisfaction Problems (CSPs). The goal in sparsification is to reduce the number of constraints in a problem instance without changing the answer, such that a bound on the number of resulting constraints can be given in terms of the number of variables~$n$. We investigate how the worst-case sparsification size depends on the types of constraints allowed in the problem formulation (the constraint language). Two algorithmic results are presented. The first result essentially shows that for any arity~$k$, the only constraint type for which no nontrivial sparsification is possible has exactly one falsifying assignment, and corresponds to logical OR (up to negations). Our second result concerns linear sparsification, that is, a reduction to an equivalent instance with~$\Oh(n)$ constraints. Using linear algebra over rings of integers modulo prime powers, we give an elegant necessary and sufficient condition for a constraint type to be captured by a degree-$1$ polynomial over such a ring, which yields linear sparsifications. The combination of these algorithmic results allows us to prove two characterizations that capture the optimal sparsification sizes for a range of Boolean CSPs. For NP-complete Boolean CSPs whose constraints are \emph{symmetric} (the satisfaction depends only on the number of \textsc{1} values in the assignment, not on their positions), we give a complete characterization of which constraint languages allow for a linear sparsification. For Boolean CSPs in which every constraint has arity at most three, we characterize the optimal size of sparsifications in terms of the largest OR that can be expressed by the constraint language.
\end{abstract}

\section{Introduction}
\paragraph*{Background}
The framework of constraint satisfaction problems (CSPs) provides a unified way to study the computational complexity of a wide variety of combinatorial problems such as \textsc{CNF-Satisfiability}, \textsc{Graph Coloring}, and \textsc{Not-All-Equal SAT}. The framework uncovers algorithmic approaches that simultaneously apply to several problems, and also identifies common sources of intractability.
For the purposes of this discussion, a CSP is specified using a (finite) \emph{constraint language},
which is a set of (finite) relations;
the problem is to decide the satisfiability of a set of constraints, where each constraint
has a relation coming from the constraint language.
The fact that many problems can be viewed as CSPs motivates the following investigation:
how does the  complexity of a CSP depend its constraint language?
%, i.e., on the types of constraints that are allowed in an instance?
A key result in this area is Schaefer's dichotomy theorem~\cite{Schaefer78}, which classifies
each CSP over the Boolean domain as polynomial-time solvable or NP-complete.

Continuing a recent line of investigation~\cite{JansenP2016,JansenP2017,LagerkvistW17}, we aim to understand
for which NP-complete CSPs an instance
can be \emph{sparsified} in polynomial time, without changing the answer.
In particular, we investigate the following questions.
Can the number of constraints be reduced to a small function of the number of variables~$n$? How does the sparsifiability of a CSP depend on its constraint language? We utilize the framework of kernelization~\cite{CyganFKLMPPS15,DowneyF13,LokshtanovMS12}, originating in parameterized complexity theory, to answer such questions.

The first results concerning polynomial-time sparsification in terms of the number~$n$ of variables or vertices were mainly negative. Under the assumption that \notcontainment (which we tacitly assume throughout this introduction), Dell and van Melkebeek~\cite{DellM14} proved a strong lower bound: For any integer~$d \geq 3$ and positive real~$\varepsilon$, there cannot be a polynomial-time algorithm that compresses any instance~$\varphi$ of \textsc{$d$-CNF-SAT} on~$n$ variables, into an equivalent SAT instance~$\varphi'$ of bitsize~$\Oh(n^{d - \varepsilon})$. In fact, there cannot even be an algorithm that transforms such~$\varphi$ into small equivalent instances~$\psi$ of an \emph{arbitrary} decision problem. Since an instance of \textsc{$d$-CNF-SAT} has at most~$2^d n^d \in \Oh(n^d)$ distinct clauses, it can \emph{trivially} be sparsified to~$\Oh(n^d)$ clauses by removing duplicates, and can be compressed to size~$\Oh(n^d)$ by storing it as a bitstring indicating for each possible clause whether or not it is present. The cited lower bound therefore shows that the trivial sparsification for \textsc{$d$-CNF-SAT} cannot be significantly improved; we say that the problem does not admit nontrivial (polynomial-time) sparsification. Following these lower bounds for SAT, a number of other results were published~\cite{DellM12,Jansen15,KratschPR16} proving other problems do not admit nontrivial sparsification either.

This pessimistic state of affairs concerning nontrivial sparsification algorithms changed several years ago, when a subset of the authors~\cite{JansenP2017} showed that the \textsc{$d$-Not-All-Equal SAT} problem \emph{does} have a nontrivial sparsification. In this problem, clauses have size at most~$d$ and are satisfied if the literals do not all evaluate to the same value. While there can be~$\Omega(n^d)$ different clauses in an instance, there is an efficient algorithm that finds a subset of~$\Oh(n^{d-1})$ clauses that preserves the answer, resulting in a compression of bitsize~$\Oh(n^{d-1} \log n)$. The first proof of this result was based on an ad-hoc application of a theorem of Lov\'asz~\cite{Lovasz76}. Later, the underlying proof technique was extracted and applied to a wider range of problems~\cite{JansenP2016}. This led to the following understanding: if each relation in the constraint language
% (i.e., all types of allowed constraints)
can be represented by a polynomial of degree at most~$d$, in a certain technical sense, then this allows the number of constraints in an $n$-variable instance of such a CSP to be reduced to~$\Oh(n^d)$. The sparsification for \textsc{$d$-Not-All-Equal SAT} is then explained by noting that such constraints can be captured by polynomials of degree~$d-1$. It is therefore apparent that finding a low-degree polynomial to capture the constraints of a CSP is a powerful tool to obtain sparsification algorithms for it. Finding such polynomials of a certain degree~$d$, or determining that they do not exist, proved a challenging and time-intensive task (cf.~\cite{JansenP17}).

The polynomial-based framework~\cite{JansenP2016} also resulted in some \emph{linear} sparsifications. Since ``1-in-$d$'' constraints (to satisfy a clause, \emph{exactly one} out of its~$\leq d$ literals should evaluate to true) can be captured by \emph{linear} polynomials, the \textsc{$1$-in-$d$-SAT} problem has a sparsification with~$\Oh(n)$ constraints for each constant~$d$. This prompted a detailed investigation into linear sparsifications for CSPs by Lagerkvist and Wahlstr{\"{o}}m~\cite{LagerkvistW17}, who used the toolkit of universal algebra in an attempt to obtain a characterization of the Boolean CSPs with a linear sparsification. Their results give a necessary and sufficient condition on the constraint language of a CSP for having a so-called Maltsev embedding over an infinite domain. They also show that when a CSP has a Maltsev embedding over a \emph{finite} domain, then this can be used to obtain a linear sparsification. Alas, it remains unclear whether Maltsev embeddings over infinite domains can be exploited algorithmically, and %therefore
a characterization of the linearly-sparsifiable CSPs is currently not known.

\paragraph*{Our contributions}
We analyze and demonstrate the power of the polynomial-based framework for sparsifying CSPs using universal algebra, linear algebra over rings, and relational analysis. We present two new algorithmic results. These allow us to characterize the sparsifiability of Boolean CSPs in two settings,
wherein we show that the polynomial-based framework yields \emph{optimal} sparsifications. In comparison to previous work~\cite{JansenP2016}, our results are much more fine-grained and based on a deeper understanding of the reasons why a certain CSP \emph{cannot} be captured by low-degree polynomials.

\subparagraph{Algorithmic results}
Our first result
(Section~\ref{sec:trivial-vs-nontrivial})
shows that, contrary to the pessimistic picture that arose during the initial investigation of sparsifiability, the phenomenon of nontrivial sparsification is widespread and occurs for almost all Boolean CSPs! We prove that if~$\Gamma$ is a constraint language whose largest constraint has arity~$k$, then the \emph{only} reason that CSP($\Gamma$) does not have a nontrivial sparsification, is that it contains an arity-$k$ relation that is essentially the $k$-ary OR (up to negating variables). When~$R \subseteq \{0,1\}^k$ is a relation with~$|\{0,1\}^k \setminus R| \neq 1$ (the number of assignments that fail to satisfy the constraint is not equal to~$1$), then it can be captured by a polynomial of degree~$k-1$. This yields a nontrivial sparsification compared to the~$\Omega(n^k)$ distinct applications of this constraint that can be in such an instance.

Our second algorithmic result
(Section~\ref{sec:balanced-to-linear})
concerns the power of the polynomial-based framework for obtaining linear sparsifications. We give a necessary and sufficient condition for a relation to be captured by a degree-1 polynomial. Say that a Boolean relation~$R \subseteq \{0,1\}^k$
is \emph{balanced} if there is \emph{no} sequence of vectors~$s_1, \ldots, s_{2n}, s_{2n+1} \in R$ for~$n \geq 1$ such that~$s_1 - s_2 + s_3 \ldots - s_{2n} + s_{2n+1} = u \in \{0,1\}^k \setminus R$. (The same vector may appear multiple times in this sum.) In other words:~$R$ is balanced if one cannot find an odd-length sequence of vectors in~$R$ for which alternating between adding and subtracting these vectors component-wise results in a 0/1-bitvector~$u$ that is outside~$R$. For example, the binary OR relation~$\kor{2} = \{0,1\}^2 \setminus \{(0,0)\}$ is \emph{not} balanced, since~$(0,1) - (1,1) + (1,0) = (0,0) \notin \kor{2}$, but the 1-in-$3$ relation~$R_{=1} = \{(1,0,0), (0,1,0), (0,0,1)\}$ is. We prove that if a Boolean relation~$R$ is balanced, then it can efficiently be captured by a degree-1 polynomial and the number of constraints that are applications of this relation can be reduced to~$\Oh(n)$. Hence when \emph{all} relations in a constraint language~$\Gamma$ are balanced---we call such a constraint language \emph{balanced}---then CSP($\Gamma$) has a sparsification with~$\Oh(n)$ constraints.
We also show that, on the other hand, if a Boolean relation~$R$ is not balanced, then there does not exist a degree-1 polynomial over any ring that captures~$R$ in the sense required for application of the polynomial framework.
The property of being balanced is (as defined)
a universal-algebraic property; these results
thus tightly bridge universal algebra and the polynomial framework.

\subparagraph{Characterizations} The property of being balanced gives an easy way to prove that certain Boolean CSPs admit linear sparsifications. But perhaps more importantly, this characterization constructively exhibits a certain \emph{witness} when a relation can \emph{not} be captured by a degree-1 polynomial, in the form of the alternating sum of satisfying assignments that yield an unsatisfying assignment. In several scenarios, we can turn this witness structure against degree-1 polynomials into a lower bound proving that the problem does not have a linear sparsification. As a consequence, we can prove two fine-grained characterizations of sparsification complexity.

\textbf{Characterization of symmetric CSPs with a linear sparsification (Section~\ref{sec:symmetric})} We say that a Boolean relation is \emph{symmetric} if the satisfaction of a constraint only depends on the \emph{number} of $1$-values taken by the variables (the \emph{weight} of the assignment), but does not depend on the \emph{positions} where these values appear. For example, ``$1$-in-$k$''-constraints are symmetric, just as ``not-all-equal''-constraints, but the relation~$R_{a\rightarrow b} = \{(0,0), (0,1), (1,1)\}$ corresponding to the truth value of~$a \rightarrow b$ is not. We prove that if a symmetric Boolean relation~$R$ is not balanced, then it can implement (Definition~\ref{def:cone}) a binary OR using constants and negations but without having to introduce fresh variables. Building on this, we prove that if such an unbalanced symmetric relation~$R$ occurs in a constraint language~$\Gamma$ for which CSP($\Gamma$) is NP-complete, then CSP($\Gamma$) does \emph{not} admit a sparsification of size~$\Oh(n^{2-\varepsilon})$ for any $\varepsilon > 0$. Consequently, we obtain a characterization of the sparsification complexity of NP-complete Boolean CSPs whose constraint language consists of symmetric relations: there is a linear sparsification if and only if the constraint language is balanced. This yields linear sparsifications in several new scenarios that were not known before.

\textbf{Characterization of sparsification complexity for CSPs
of low arity (Section~\ref{sec:arity-3})}
By combining the linear sparsifications guaranteed by balanced constraint languages with the nontrivial sparsification when the largest-arity relations do not have exactly one falsifying assignment, we obtain an exact characterization of the optimal sparsification size for all Boolean CSPs where each relation has arity at most three. For a Boolean constraint language~$\Gamma$ consisting of relations of arity at most three, we characterize the sparsification complexity of~$\Gamma$ as an integer~$k \in \{1,2,3\}$ that represents the largest OR that~$\Gamma$ can implement using constants and negations, but without introducing fresh variables. Then we prove that CSP($\Gamma$) has a sparsification of size~$\Oh(n^k)$, but no sparsification of size~$\Oh(n^{k-\varepsilon})$ for any~$\varepsilon > 0$, giving matching upper and lower bounds. Hence for all Boolean CSPs with constraints of arity at most three, the polynomial-based framework gives provably \emph{optimal} sparsifications.

\section{Preliminaries}
\label{sec:preliminaries}

For a positive integer $n$, define $[n] := \{1,2,\ldots,n\}$.
For an integer $q$, we let $\mathbb{Z}/q\mathbb{Z}$ denote the integers modulo $q$. These form a field if $q$ is prime, and a ring otherwise. We will use $x \equiv_q y$ to denote that $x$ and $y$ are congruent modulo $q$, and $x\not\equiv_q y$ to denote that they are incongruent modulo $q$. For statements marked with a star $(\bigstar)$, the (full) proof can be found in Appendix \ref{app:missing-proofs}.

\subparagraph{Parameterized complexity}

A parameterized problem~$\mathcal{Q}$ is a subset of~$\Sigma^* \times \mathbb{N}$, where~$\Sigma$ is a finite alphabet.
%\begin{definition}
Let~$\mathcal{Q},\mathcal{Q}'\subseteq\Sigma^*\times\mathbb{N}$ be parameterized problems and let~$h\colon\mathbb{N}\to\mathbb{N}$ be a computable function. A \emph{generalized kernel for~$\mathcal{Q}$ into~$\mathcal{Q}'$ of size~$h(k)$} is an algorithm that, on input~$(x,k)\in\Sigma^*\times\mathbb{N}$, takes time polynomial in~$|x|+k$ and outputs an instance~$(x',k')$ such that:
%\begin{enumerate}
%\item
(i) $|x'|$ and~$k'$ are bounded by~$h(k)$, and
%\item
(ii) $(x',k') \in \mathcal{Q}'$ if and only if~$(x,k) \in \mathcal{Q}$.
%\end{enumerate}
The algorithm is a \emph{kernel} for~$\mathcal{Q}$ if~$\mathcal{Q}'=\mathcal{Q}$. %It is a \emph{polynomial (generalized) kernel} if~$h(k)$ is a polynomial.
%\end{definition}

Since a polynomial-time reduction to an equivalent sparse instance yields a generalized kernel, lower bounds against generalized kernels can be used to prove the non-existence of such sparsification algorithms. To relate the sparsifiability of different problems to each other, the following notion is useful.

\begin{definition}\label{def:lpt}
Let $\mathcal{P}, \mathcal{Q} \subseteq \Sigma^* \times \mathbb{N}$ be two parameterized problems.
A \emph{linear-parameter transformation} from
$\mathcal{P}$
to $\mathcal{Q}$ is a polynomial-time algorithm that, given an instance $(x,k) \in \Sigma^* \times \mathbb{N}$ of $\mathcal{P}$, outputs an instance $(x',k') \in \Sigma^* \times \mathbb{N}$ of $\mathcal{Q}$ such that the following holds:
\begin{enumerate}
\item  $(x,k) \in \mathcal{Q}$ if and only if $(x',k') \in \mathcal{P}$, and
\item $k' \in \mathcal{O}(k)$.
\end{enumerate}
\end{definition}

\noindent It is well-known~\cite{BodlaenderJK14,BodlaenderTY11} that the existence of a linear-parameter transformation from problem~$\mathcal{P}$ to~$\mathcal{Q}$ implies that any generalized kernelization lower bound for~$\mathcal{P}$, also holds for~$\mathcal{Q}$.

\subparagraph{Operations, relations, and preservation}

A \emph{Boolean operation} is a mapping from $\{ 0, 1 \}^k$
to $\{ 0, 1 \}$,
where $k$, a natural number,
is said to be the \emph{arity} of the operation;
we assume throughout that operations have positive arity.
%Unless mentioned otherwise, we assume each operation
%under discussion to have arity greater than or equal to $1$.
From here, we define a
\emph{partial Boolean operation} in the usual way,
that is, it is a mapping
from a subset of $\{ 0, 1 \}^k$
to $\{ 0, 1 \}$.
We say that a partial Boolean operation $f$
of arity $k$
is \emph{idempotent} if $f(0,\ldots,0) = 0$
and $f(1,\ldots,1) = 1$; and,
\emph{self-dual}
if for all $(a_1,\ldots,a_k) \in \{ 0, 1 \}^k$,
when $f(a_1,\ldots,a_k)$ is defined,
it holds that $f(\neg a_1,\ldots, \neg a_k)$ is defined
and $f(a_1,\ldots,a_k) = \neg f(\neg a_1,\ldots,\neg a_k)$.

\begin{definition}\label{def:balanced_op}
A partial Boolean operation $f \colon \{ 0, 1 \}^k \to \{ 0, 1 \}$
is \emph{balanced} if there exist
integer values $\alpha_1, \ldots, \alpha_k$,
called the \emph{coefficients} of $f$,
such that
\begin{itemize}

\item $\sum_{i\in [k]} \alpha_i = 1$,

\item $(x_1,\ldots,x_k)$ is in the domain of $f$ if and only if
$\sum_{i \in [k]} \alpha_i x_i \in \{0,1\}$, and

\item  $f(x_1,\ldots,x_k) = \sum_{i \in [k]}\alpha_i x_i$ for all tuples in its domain.

\end{itemize}
\end{definition}

A \emph{relation} over the set $D$ is a subset of $D^k$;
here, $k$ is a natural number
called the \emph{arity} of the relation.  Throughout, we assume that each relation is over
a finite set $D$.
A \emph{Boolean relation} is a relation over $\{ 0, 1 \}$.

\begin{definition}
For each $k \geq 1$, we use $\kor{k}$ to denote the relation $\{ 0, 1 \}^k \setminus \{ (0,\ldots,0) \}$.
\end{definition}

A \emph{constraint language over $D$}
is a finite set of relations over $D$;
a \emph{Boolean constraint language} is a constraint language over $\{ 0, 1 \}$. For a Boolean constraint language $\Gamma$, we define $\csp(\Gamma)$ as follows.

\defparproblem{$\csp(\Gamma)$}{A tuple $(\C,V)$, where $\C$ is a finite set of constraints,
$V$ is a finite set of variables, and
each constraint is a pair $R(x_1,\ldots,x_k)$ for $R \in \Gamma$ and $x_1,\ldots,x_k \in V$. }{The number of variables $|V|$.}{Does there exist a \emph{satisfying assignment},
that is, an assignment $f \colon V \rightarrow \{0,1\}$ such that for each constraint $R(x_1,\ldots,x_k) \in \C$ it holds that $(f(x_1),\ldots,f(x_k)) \in R$?}

Let $f \colon \{ 0, 1 \}^k \to \{ 0, 1 \}$ be a partial Boolean operation,
and let $T \subseteq \{ 0, 1 \}^n$ be a Boolean relation.
We say that $T$ is \emph{preserved} by $f$ when,
for any tuples $t^1 = (t^1_1, \ldots, t^1_n), \ldots, t^k = (t^k_1, \ldots, t^k_n) \in T$,
if all entries of the tuple
$(f(t^1_1, \ldots, t^k_1), \ldots, f(t^1_n, \ldots, t^k_n))$ are defined,
then this tuple is in $T$.
We say that a Boolean constraint language $\Gamma$ is \emph{preserved} by $f$
if each relation in $\Gamma$ is preserved by $f$.
We say that a Boolean relation is \emph{balanced} if it is preserved by all balanced operations,
and that a Boolean constraint language is \emph{balanced} if each relation therein is balanced.

Define an \emph{alternating operation} to be a balanced operation
$f \colon \{ 0, 1 \}^k \to \{ 0, 1 \}$ such that $k$ is odd and the coefficients alternate
between $+1$ and $-1$, so that $\alpha_1 = +1$, $\alpha_2 = -1$, $\alpha_3 = +1$, $\ldots$, $\alpha_k = +1$.
We have the following.

\begin{proposition}[$\bigstar$]
\label{prop:balanced-acts-as-alternating}
%A relation is balanced if and only if it is preserved by all alternating operations.
A Boolean relation~$R$ is balanced if and only if for all odd~$k \geq 1$, the relation $R$ is preserved by the alternating operation of arity~$k$.
\end{proposition}

We will use the following straightforwardly verified fact
tacitly, throughout.

\begin{observation}
Each balanced operation is idempotent and self-dual.
\end{observation}

For $b \in \{ 0, 1 \}$,
let $u_b \colon \{ 0, 1 \} \to \{ 0, 1 \}$ be the unary operation defined by
$u_b(0) = u_b(1) = b$;
let $\major \colon \{ 0, 1 \}^3 \to \{ 0, 1 \}$ to be the operation defined
by $\major(x,y,z) = (x \wedge y) \vee (x \wedge z) \vee (y \wedge z)$;
and,
let $\minor \colon \{ 0, 1 \}^3 \to \{ 0, 1 \}$ to be the operation defined by
$\minor(x,y,z) = x \oplus y \oplus z$,
where $\oplus$ denotes exclusive OR.
We say that a Boolean constraint language $\Gamma$ is \emph{tractable} if it is preserved
by one of the six following operations: $u_0$, $u_1$, $\wedge$, $\vee$, $\minor$, $\major$;
we say that $\Gamma$ is \emph{intractable} otherwise.
It is known that, in terms of classical complexity,
the problem $\csp(\Gamma)$ is polynomial-time decidable when $\Gamma$ is tractable,
and that the problem $\csp(\Gamma)$ is NP-complete when $\Gamma$ is intractable
(see \cite{Chen09-Rendezvous} for a proof; in particular, refer there to the proof of Theorem 3.21).

\subparagraph{Constraint Satisfaction and Definability}

\begin{assumption}
By default, we assume in the sequel that the operations, relations,
and constraint languages under discussion are Boolean, and
that said operations and relations are of positive arity.
We nonetheless sometimes describe them as being Boolean, for emphasis.
\end{assumption}

\begin{definition} \label{def:cone}
Let us say that a Boolean relation $T$ of arity $m$ is
\emph{cone-definable} from a Boolean relation $U$ of arity $n$
if there exists a tuple
$(y_1, \ldots, y_n)$
where:
\begin{itemize}

\item for each $j \in [n]$,
it holds that $y_j$ is an element of
$\{ 0, 1 \} \cup \{ x_1, \ldots, x_m \} \cup
\{ \neg x_1, \ldots, \neg x_m \}$;

\item for each $i \in [m]$, there exists $j \in [n]$
such that $y_j \in \{ x_i, \neg x_i \}$;
and,

\item for each $f \colon \{ x_1, \ldots, x_m \} \to \{ 0, 1 \}$,
it holds that
$(f(x_1), \ldots, f(x_m)) \in T$ if and only if
$(\hat{f}(y_1), \ldots, \hat{f}(y_n)) \in U$.
Here, $\hat{f}$ denotes the natural extension of $f$
where $\hat{f}(0) = 0$,
$\hat{f}(1) = 1$,
and $\hat{f}(\neg x_i) = \neg f(x_i)$.

\end{itemize}
(The prefix \emph{cone} indicates the allowing of
{\bf co}nstants and {\bf ne}gation.)
\end{definition}

\begin{example}
Let $R = \{ (0,0),(0,1) \}$
and let $S = \{ (0,1),(1,1) \}$.
We have that $R$ is cone-definable from $S$ via the tuple $(\neg x_2, \neg x_1)$;
also, $S$ is cone-definable from $R$ via the same tuple.
\end{example}

When $\Gamma$ is a constraint language over $D$, we use $\Gamma^*$ to denote the
expansion of $\Gamma$
where each element of $D$ appears as a relation, that is,
we define $\Gamma^*$ as $\Gamma \cup \{ \{ (d) \} ~|~ d \in D \}$.

The following is a key property of cone-definability; it states that
relations that are cone-definable from a constraint language $\Gamma$
may be simulated by the constraint language, and thus used to prove hardness results
for $\csp(\Gamma)$.

\begin{proposition}[$\bigstar$]
\label{prop:lpt-from-cone-definability}
Suppose that $\Gamma$ is an intractable constraint language,
and that $\Delta$ is a constraint language such that
each relation in $\Delta$ is cone-definable from a relation in $\Gamma$.
Then, there exists a linear-parameter transformation from
$\csp(\Gamma^* \cup \Delta)$ to $\csp(\Gamma)$.
\end{proposition}

%%%%%%%%%%%%%%

\section{Trivial versus non-trivial sparsification}
\label{sec:trivial-vs-nontrivial}

It is well known that $k$-CNF-SAT allows no non-trivial sparsification,
for each $k \geq 3$ \cite{DellM14}. This means that we cannot efficiently reduce the number of clauses in such a formula to $\Oh(n^{k-\varepsilon})$. The $k$-\OR relation is special, in the sense that there is exactly one $k$-tuple that is not contained in the relation. We show in this section that when considering $k$-ary relations for which there is more than one $k$-tuple not contained in the relation, a non-trivial sparsification is always possible. In particular, the number of constraints of any input can efficiently be reduced to $\Oh(n^{k-1})$.
Using Lemmas \ref{lem:2-unsat-implies-kernel} and \ref{lem:1-unsat-defines-OR}, we will completely classify the constraint languages that allow a non-trivial sparsification as follows.

\begin{theorem}[$\bigstar$]\label{thm:trivial-vs-nontrivial}
Let $\Gamma$ be an intractable (Boolean) constraint language. Let $k$ be the maximum arity of any relation $R \in \Gamma$. The following dichotomy holds.
\begin{itemize}
\item If for all $R \in \Gamma$ it holds that $|R| \neq 2^k-1$, then $\csp(\Gamma)$ has a kernel with $\Oh(n^{k-1})$ constraints that can be stored in $\Oh(n^{k-1}\log n)$ bits.
\item If there exists $R \in \Gamma$ with $|R| = 2^k-1$, then $\csp(\Gamma)$ has no generalized kernel of bitsize $\Oh(n^{k-\varepsilon})$ for any $\varepsilon > 0$, unless \containment.
\end{itemize}
\end{theorem}

To obtain the kernels given in this section, we will heavily rely on the following notion for representing constraints by polynomials.
\begin{definition} \label{def:capture}
Let $R$ be a $k$-ary Boolean relation. We say that a polynomial $p_u$ over a ring~$E_u$ \emph{captures} an unsatisfying assignment $u \in \{0,1\}^k\setminus R$ with respect to $R$, if the following two conditions hold over~$E_u$.
\begin{align}
&p_u(x_1,\ldots,x_k) = 0 \text{ for all } (x_1,\ldots,x_k) \in R\label{eq:sat-R}, \text{ and} \\
&p_u(u_1,\ldots,u_k) \neq 0.\label{eq:unsat-u}
\end{align}
\end{definition}

The following Theorem is a generalization of Theorem 16 in \cite{JansenP18}. The main improvement is that we now allow the usage of different polynomials, over different rings, for each $u\notin R$. Previously, all polynomials had to be given over the same ring, and each constraint was captured by a single polynomial.

\begin{theorem}[$\bigstar$]\label{thm:only-need-polynomials}
Let $R \subseteq \{0,1\}^k$ be a fixed $k$-ary relation, such that for every $u \in \{0,1\}^k\setminus R$ there exists a ring~$E_u \in \{\mathbb{Q}\}	 \cup \{\mathbb{Z}/q_u\mathbb{Z} \mid q_u \text{ is a prime power}\}$ and polynomial $p_u$ over~$E_u$ of degree at most $d$ that captures $u$ with respect to $R$. Then there exists a polynomial-time algorithm that, given a set of constraints \C over $\{R\}$ over $n$ variables, outputs $\C' \subseteq \C$ with $|\C'| = \Oh(n^{d})$, such that any Boolean assignment satisfies all constraints in $\C$ if and only if it satisfies all constraints in $\C'$.
\end{theorem}

The next lemma states that any $k$-ary Boolean relation $R$ with $|R| < 2^k - 1$ admits a non-trivial sparsification. To prove the lemma, we show that such relations can be represented by polynomials of degree at most $k-1$, such that the sparsification can be obtained using Theorem \ref{thm:only-need-polynomials}. Since relations with $|R| = 2^k$ have a  sparsification of size $\Oh(1)$, as constraints over such relations are satisfied by any assignment, it will follow that $k$-ary relations with $|\{0,1\}^k \setminus R| \neq 1$ always allow a non-trivial sparsification.

\begin{lemma}[$\bigstar$]\label{lem:2-unsat-implies-kernel}
Let $R$ be a $k$-ary Boolean relation with $|R| < 2^k-1$. Let $\C$ be a set of constraints over $\{R\}$, using $n$ variables. Then there exists a polynomial-time algorithm that outputs $\C' \subseteq \C$ with $|\C'| = \Oh(n^{k-1})$, such that a Boolean assignment satisfies all constraints in $\C'$ if and only if it satisfies all constraints in $\C$.
\end{lemma}
\begin{proof}[Proof sketch.]
%In the proof of this lemma we show that for every $u \in \{0,1\}^k \setminus R$, there exists a degree-$(k-1)$ polynomial  over the integers that captures $u$, such that the result follows  from Theorem~\ref{thm:only-need-polynomials}.
%To prove that such polynomials exist, we then do induction on $k$. In this proof sketch however, we will  give some intuition by providing some examples. Let $R$ be a $3$-ary relation, and let $u \neq w$ be such that $u,w\notin R$. We give two examples.
%
%($u = (0,0,0)$, $w =(1,1,1)$) In this case, $u_i \neq w_i$ for all $i\in[3]$. We define
%\[\textstyle p_u(x_1,x_2,x_3) := (x_1+x_2+x_3-1)\cdot (x_1+x_2+x_3-2).\]
%Verify that indeed, $p_u(0,0,0) = 2 \neq 0$ and thereby $p_u$ satisfies \eqref{eq:unsat-u}. Furthermore, for any $(x_1,x_2,x_3)\in R$, it holds that $x_1+x_2+x_3 \in \{1,2\}$ and thereby $p_u(x_1,x_2,x_3) = 0$. Thus, for any $(x_1,x_2,x_3) \in R$ we obtained $p_u(x_1,x_2,x_3) = 0$, satisfying \eqref{eq:sat-R}.
%
%($u = (0,0,0)$, $w = (0,1,1)$) Here $u_1 = w_1=0$ and we define
%\[\textstyle p_u(x_1,x_2,x_3) : = (1-x_1) \cdot (x_2+x_3 - 1).\]
%Verify that $p_u(x_1,x_2,x_3) = 0$ for any $(x_1,x_2,x_3)$ different from $u$ or $w$, as in this case $x_1=1$ or $(x_2,x_3)$ differs from both $(u_2,u_3)$ and $(w_2,w_3)$. Furthermore, $p_u(0,0,0) = -1 \neq 0$.
We will show that for every $u \in \{0,1\}^k \setminus R$, there exists a degree-$(k-1)$ polynomial~$p_u$ over $\mathbb{Q}$ that captures $u$, such that the result follows  from Theorem~\ref{thm:only-need-polynomials}.
We will prove the existence of such a polynomial by induction on $k$. For $k=1$, the lemma statement implies that $R  = \emptyset$. Thereby, for any $u \notin R$, we simply choose $p_u(x_1) := 1$. This polynomial satisfies the requirements, and has degree $0$.
Let $k > 1$ and let $u=(u_1,\ldots,u_k) \in \{0,1\}^k \setminus R$. Since $|R| < 2^k-1$, we can choose $w = (w_1,\ldots,w_k)$ such that $w \in \{0,1\}^k \setminus R$ and $w \neq {u}$. We distinguish two cases, depending on whether~$u$ and~$w$ agree on some position.

Suppose $u_i \neq w_i$ for all $i$, and assume for concreteness that~$u = (0, \ldots, 0)$ and~$w = (1, \ldots, 1)$. Then the polynomial~$p_u(x_1,\ldots,x_k) := \prod_{i=1}^{k-1} (i - \sum _{j=1}^{k} x_j)$ suffices: $p_u(0,\ldots,0) = \prod_{j=1}^{k-1} j \neq 0$, while for any $(x_1,\ldots,x_k)\in R$, it holds that $\sum_{i=1}^k x_i \in [k-1]$ and thereby $p_u(x_1,\ldots,x_k) = 0$; the product has a $0$-term. Other values of~$u$ and~$w$ are handled similarly.

Now suppose $u_i = w_i$ for some $i \in [k]$, and assume for concreteness that~$u_1 = w_1 = 1$. Define $R' := \{(x_2,\ldots,x_k) \mid (1,x_2,\ldots,x_k)\in R\}$ and let $u' := (u_2, \ldots, u_k)$. Since~$(u_2, \ldots, u_k)$ and~$(w_2, \ldots, w_k)$ are distinct tuples not in~$R'$, by induction there is a polynomial $p_{u'}$ of degree~$k-2$ that captures $u'$ with respect to $R'$. Then the polynomial $p_u(x_1,\ldots,x_k) := x_1 \cdot p_{u'}(x_2,\ldots,x_k)$ has degree~$k-1$ and captures~$u$ with respect to~$R$.
\end{proof}

To show the other part of the dichotomy, we will need the following theorem.
\begin{theorem}[$\bigstar$]\label{thm:LB-by-defining-OR}
Let $\Gamma$ be an intractable (Boolean) constraint language,
and let $k \geq 1$.
If there exists $R \in \Gamma$ such that $R$ cone-defines $k$-\OR, then $\csp(\Gamma)$ does not have a generalized kernel of size $\Oh(n^{k-\varepsilon})$, unless \containment.
\end{theorem}

The next lemma formalizes the idea that any $k$-ary relation with $|\{0,1\}^k\setminus R| = 1$ is equivalent to $k$-\OR, up to negation of variables. The proof of the dichotomy given in Theorem \ref{thm:trivial-vs-nontrivial} will follow from Lemma \ref{lem:2-unsat-implies-kernel}, together with the next lemma and Theorem \ref{thm:LB-by-defining-OR}.
\begin{lemma}[$\bigstar$]\label{lem:1-unsat-defines-OR}
Let $R$ be a $k$-ary relation with $|R| = 2^k-1$. Then $R$ cone-defines $k$-\OR.
\end{lemma}

\iffalse
Using the above results, we give a proof sketch of Theorem \ref{thm:trivial-vs-nontrivial}.

\begin{proof}[Proof sketch of Theorem \ref{thm:trivial-vs-nontrivial}]
Suppose that for all $R \in \Gamma$, it holds that $|R| \neq 2^k-1$. Given an instance of $\csp(\Gamma)$,  a kernel can be obtained by reducing the number of constraints of the form $R(x_1,\ldots,x_\ell)$ to $\Oh(n^{k-1})$ for all $R \in \Gamma$, which can be done by applying Lemma \ref{lem:2-unsat-implies-kernel}. Since $|\Gamma|$ is fixed, this gives the desired kernel.

Suppose that there exists $R \in \Gamma$ with $|R| = 2^k-1$. It follows from Lemma \ref{lem:1-unsat-defines-OR} that $R$ cone-defines $k$-\OR. Since $\Gamma$ is intractable,  Theorem \ref{thm:LB-by-defining-OR} implies that $\csp(\Gamma)$ has no generalized kernel of size $\Oh(n^{k-\varepsilon})$, unless \containment.
\end{proof}
\fi
\section{From balanced operations to linear sparsification}
\label{sec:balanced-to-linear}
The main result of this section is the following theorem, which we prove below.

\begin{theorem}\label{thm:preserved_by_balanced}
Let $\Gamma$ be a balanced (Boolean) constraint language.
% such that all $R\in \Gamma$ are preserved by all balanced operations.
Then $\csp(\Gamma)$ has a kernel with $\Oh(n)$ constraints that are a subset of the original constraints. The kernel can be stored using $\Oh(n\log n)$ bits.
\end{theorem}

To prove the theorem, we will use two additional technical lemmas. To state them, we introduce some notions from linear algebra.
 Given a set $S = \{s_1,\ldots,s_n\}$ of $k$-ary vectors in $\mathbb{Z}^k$, we define $\text{span}_{\mathbb{Z}}(S)$ as the set of all vectors $y$ in $\mathbb{Z}^k$ for which there exist $\alpha_1,\ldots,\alpha_n \in \mathbb{Z}$ such that $y = \sum_{i\in[n]} \alpha_is_i$. Similarly, we define $\text{span}_q(S)$ as the set of all $k$-ary vectors $y$ over $\mathbb{Z}/q\mathbb{Z}$, such that there exist $\alpha_1,\ldots,\alpha_n$ such that  $y \equiv_q \sum_{i\in[n]} \alpha_is_i$. For an $m\times n$ matrix $S$, we use $s_i$ for $i\in[m]$ to denote the $i$'th row of $S$.

\begin{lemma}[$\bigstar$]\label{lem:prime-powers-to-int}
Let $S$ be an $m\times n$ integer matrix. Let ${u}\in \mathbb{Z}^n$ be a row vector.
If $u \in \text{span}_q(\{s_1,\ldots,s_m\})$ for all prime powers $q$, then $u \in \text{span}_{\mathbb{Z}}(\{s_1,\ldots,s_m\})$.
\end{lemma}

\begin{lemma}[$\bigstar$]\label{lem:no-solution-implies-span}
Let $q$ be a prime power. Let $A$ be an $m \times n$ matrix over $\mathbb{Z}/q\mathbb{Z}$. Suppose there exists no constant $c\not\equiv_q 0$ for which the system
$Ax\equiv_q b$
has a solution, where $b := (0,\ldots,0,c)^T$ is the vector with $c$ on the last position and zeros in all other positions.

Then $a_m \in \text{span}_q(\{a_1,\ldots,a_{m-1}\})$.
\end{lemma}

Using these tools from linear algebra, we now prove the main sparsification result.

\begin{proof}[Proof of Theorem \ref{thm:preserved_by_balanced}]
We show that for all relations $R$ in the balanced constraint language~$\Gamma$, for all $u \notin R$, there exists a linear polynomial $p_u$ over a ring $E_u \in \{\mathbb{Z}/q_u\mathbb{Z}\mid q_u \text{ is a prime power}\}$ that captures $u$ with respect to $R$. By applying Theorem \ref{thm:only-need-polynomials} once for each relation~$R \in \Gamma$, to reduce the number of constraints involving~$R$ to~$\Oh(n)$, we then reduce any $n$-variable instance of $\csp(\Gamma$) to an equivalent one on~$|\Gamma| \cdot \Oh(n) \in \Oh(n)$ constraints.

Suppose for a contradiction that there exists $R \in \Gamma$ and $u \notin R$, such that no prime power $q$ and polynomial $p$ over $\mathbb{Z}/q\mathbb{Z}$ exist that satisfy conditions \eqref{eq:sat-R} and \eqref{eq:unsat-u}. We can view the process of finding such a linear polynomial, as solving a set of linear equations whose unknowns are the coefficients of the polynomial. We have a linear equation for each evaluation of the polynomial for which we want to enforce a certain value.

Let $R = \{r_1,\ldots,r_\ell\}$. By the non-existence of $p$ and $q$, the system
\[
\left(
\begin{matrix}
1 & r_{1,1} & r_{1,2} & \ldots & r_{1,k}\\
1 & r_{2,1} & r_{2,2} & \ldots & r_{2,k}\\
\vdots & \vdots & \vdots & \ddots & \vdots \\
1 & r_{\ell,1} & r_{\ell,2} & \ldots & r_{\ell,k}\\
1 & u_1 & u_2 & \ldots & u_k
\end{matrix}
\right)
\left(
\begin{matrix}
\alpha_0\\
\alpha_1\\
\alpha_2\\
\vdots\\
\alpha_k
\end{matrix}
\right)\equiv_{q}
\left(
\begin{matrix}
0\\
0\\
\vdots\\
0\\
c
\end{matrix}
\right)
\]
has no solution for any prime power $q$ and $c \not\equiv_q 0$. Otherwise, it is easy to verify that $q$ is the desired prime power and $p(x_1,\ldots,x_k):= \alpha_0 + \sum_{i=1}^k \alpha_i x_i$ is the desired polynomial.

The fact that no solution exists, implies that $(1,u_1,\ldots,u_k)$ is in the span of the remaining rows of the matrix, by Lemma \ref{lem:no-solution-implies-span}. But this implies that for any prime power $q$, there exist coefficients $\beta_1,\ldots,\beta_\ell$ over $\mathbb{Z}/q\mathbb{Z}$ such that ${u} \equiv_{q} \sum \beta_i r_i$. Furthermore, since the first column of the matrix is the all-ones column, we obtain that $\sum \beta_i \equiv_{q} 1$. By Lemma \ref{lem:prime-powers-to-int}, it follows that there exist integer coefficients $\gamma_1,\ldots,\gamma_\ell$ such that $\sum \gamma_i = 1$  and furthermore ${u}= \sum \gamma_i r_i$. But it immediately follows that $R \in \Gamma$ is not preserved by the balanced operation given by $f(x_1,\ldots,x_\ell) := \sum \gamma_ix_i$, which contradicts the assumption that~$\Gamma$ is balanced.
\end{proof}

The kernelization result above is obtained by using the fact that when $\Gamma$ is balanced, the constraints in $\csp(\Gamma)$ can be replaced by linear polynomials. We show in the next theorem that this approach fails when $\Gamma$ is not balanced.% This unfortunately does not immediately imply a lower bound.

\begin{theorem}[$\bigstar$]\label{thm:not-balanced-no-polynomial}
Let $R$ be a $k$-ary relation that is not balanced. Then there exists $u \in \{0,1\}^k \setminus R$ for which there exists no polynomial $p_u$ over any ring~$E$ that captures $u$ with respect to $R$.
\end{theorem}

\section{Characterization of symmetric CSPs with linear sparsification}
\label{sec:symmetric}

In this section, we characterize the symmetric constraint languages $\Gamma$ for which $\csp(\Gamma)$ has a linear sparsification.

\begin{definition}
We say a $k$-ary Boolean relation $R$ is \emph{symmetric}, if there exists $S \subseteq \{0,1,\ldots,k\}$ such that a tuple $x = (x_1,\ldots,x_k)$ is in $R$ if and only if
$\weight(x) \in S$. We call $S$ the set of \emph{satisfying weights} for $R$.
\end{definition}
We will say that a constraint language $\Gamma$ is symmetric, if it only contains symmetric relations. We will prove the following theorem at the end of this section.

\begin{theorem}\label{thm:classification:symmetric}
Let $\Gamma$ be a finite Boolean symmetric intractable constraint language.
\begin{itemize}
\item If $\Gamma$ is balanced, then $\csp(\Gamma)$ has a kernel with $\Oh(n)$ constraints that can be stored in $\Oh(n\log n)$ bits.
\item If $\Gamma$ is  not balanced, then $\csp(\Gamma)$ does not have a generalized kernel of size $\Oh(n^{2-\varepsilon})$ for any $\varepsilon > 0$, unless \containment.
\end{itemize}
\end{theorem}

%Let the \emph{weight} of a tuple $x\in \{0,1\}^k$,
%denoted by $\w{x}$, be defined as $\sum_{i\in[k]}x_i$.

%We define the \emph{weight} of a tuple $t = (t_1, \ldots, t_n)$,
%denoted by $\weight(t)$,
%as the number of entries equal to $1$, or equivalently,
%as $\sum_{i=1}^n t_i$, where the sum operation is
%that of the integers.

To show this, we use the following lemma.

\begin{lemma}[$\bigstar$]\label{lem:vector-to-int}
Let $R$ be a $k$-ary symmetric relation with satisfying weights $S \subseteq \{0,1,\ldots,k\}$. Let $U:= \{0,1,\ldots,k\} \setminus S$. If there exist $a,b,c \in S$ and $d \in U$ such that $a-b+c=d$, then $R$ cone-defines $2$-\OR.
\end{lemma}
\begin{proof}[Proof sketch]
We will demonstrate the result in the case that $b \leq a$, $b \leq c$, and $b \leq d$;
the other cases are similar. We use the following tuple to express $x_1\vee x_2$.
\[(\underbrace{\neg x_1,\ldots,\neg x_1}_{(a-b)\text{ copies}}, \underbrace{\neg x_2,\ldots,\neg x_2}_{(c-b)\text{ copies}},\underbrace{1,\ldots,1}_{b\text{ copies}},\underbrace{0,~\ldots~,0}_{(k-d)\text{ copies}}).\]
Let $f \colon \{x_1,x_2\}\rightarrow \{0,1\}$, then $(\neg f(x_1),\ldots,\neg f(x_1),\neg f(x_2),\ldots,\neg f(x_2),1,\ldots,1,0,\ldots,0)$ has weight $d \notin S$ when $f(x_1) = f(x_2) = 0$. It is easy to verify that in all other cases, the weight is one of $a, b, c \in S$ and hence the tuple belongs to~$R$. The other cases are similar.
\end{proof}

We now give the main lemma that is needed to prove Theorem \ref{thm:classification:symmetric}. It shows that if a relation is symmetric and not balanced, it must cone-define $2$-\OR.

\begin{lemma}\label{lem:symmetric:lb}
Let $R$ be a symmetric (Boolean) relation of arity $k$. If $R$ is not balanced, then $R$ cone-defines $2$-\OR.
\end{lemma}
\begin{proof}
Let $f$ be a  balanced operation that does not preserve $R$.
Since $f$ has integer coefficients, it follows that there exist (not necessarily distinct) $r_1,\ldots,r_m \in R$, such that $r_1 - r_2 + r_3 - r_4 \cdots + r_m = u$ for some  $u \in \{0,1\}^k \setminus R$ and odd~$m \geq 3$. Thereby, $\weight(r_1) - \weight(r_2) + \weight(r_3)-\weight(r_4) \cdots + \weight(r_m) = \weight(u)$. Let $S$ be the set of satisfying weights for $R$ and let $U := \{0,\ldots,k\}\setminus S$. Define $s_i := \weight(r_i)$ for $i \in [m]$, and $t =\weight(u)$, such that $s_1 - s_2 + s_3 - s_4 \ldots + s_m = t$, and furthermore $s_i \in S$ for all $i$, and $t \in U$.  We show that there exist $a,b,c \in S$ and $d\in U$ such that $a - b + c = d$, such that the result follows from Lemma \ref{lem:vector-to-int}. We do this by induction on the length of the alternating sum.

If $m=3$, we have that $s_1 - s_2 + s_3 = t$ and define $a:= s_1$, $b:=s_2$, $c:=s_3$, and $d:=t$.

If $m > 3$, we will use the following claim.
\begin{claim}[$\bigstar$]\label{claim:reduce-to-three}
Let $s_1,\ldots,s_m \in S$ and $t \in U$ such that $s_1 - s_2 + s_3 - s_4\dots + s_m = t$.
There exist distinct {$i,j,\ell \in [m]$} with $i,j$ odd and $\ell$ even, such that $s_i - s_\ell + s_j \in \{0,\ldots,k\}$.
\end{claim}

Use Claim \ref{claim:reduce-to-three} to  find $i,j,\ell$ such that $s_i - s_\ell + s_j \in \{0,\ldots,k\}$. We consider two options. If $s_i - s_\ell + s_j\in U$, then define $d := s_i - s_\ell + s_j$, $a := s_i$, $b:=s_\ell$, and $c := s_j$ and we are done.  The other option is that $s_i - s_\ell + s_j=s\in S$.
Replacing $s_i - s_\ell + s_j$ by $s$ in  $s_1 - s_2 + s_3 - s_4 \cdots + s_m$ gives a shorter alternating sum with result $t$. We obtain $a,b,c$, and $d$ by the induction hypothesis.

Thereby, we have obtained $a,b,c\in S$, $d \in U$ such that $a-b+c=d$. It now follows from Lemma \ref{lem:vector-to-int} that $R$ cone-defines $2$-\OR.
\end{proof}
Using the lemma above, we can now prove Theorem \ref{thm:classification:symmetric}.
\begin{proof}[Proof of Theorem \ref{thm:classification:symmetric}]
If $\Gamma$ is balanced, it follows from Theorem \ref{thm:preserved_by_balanced} that $\csp(\Gamma)$ has a kernel with $\Oh(n)$ constraints that can be stored in $\Oh(n\log n)$ bits. Note that the assumption that $\Gamma$ is symmetric is not needed in this case.

If the symmetric constraint language $\Gamma$ is not balanced, then~$\Gamma$ contains a symmetric relation $R$ that is not balanced. It follows from Lemma \ref{lem:symmetric:lb} that $R$ cone-defines the $2$-\OR relation. Thereby, we obtain from Theorem \ref{thm:LB-by-defining-OR} that  $\csp(\Gamma)$ has no generalized kernel of size $\Oh(n^{2-\varepsilon})$ for any $\varepsilon > 0$, unless \containment.
\end{proof}

\section{Low-arity classification}
\label{sec:arity-3}
%All relations and operations in this section
%will be Boolean; however, we sometimes identify them
%as Boolean for emphasis.

%\subsection{Definability}

%\subsection{Linear versus non-linear sparsification}
In this section, we will give a full classification of the sparsifiability for constraint languages that consist only of low-arity relations. The next results will show that in this case, if the constraint language is not balanced, it can cone-define the $2$-\OR relation.

%The following fact is straightforward to verify, by using
%the fact that each balanced operation is idempotent.

\begin{observation}
\label{obs:Boolean-arity-one}
Each relation of arity $1$ is balanced.
\end{observation}

\begin{theorem}[$\bigstar$]
\label{thm:Boolean-arity-two}
A relation of arity $2$ is balanced
if and only if it is
not cone-interdefinable with the $\kor{2}$ relation.
%it has size $3$.
\end{theorem}

\begin{theorem}[$\bigstar$]\label{thm:not-balanced-implies-or}
Suppose that $U \subseteq \{ 0, 1 \}^3$ is an arity $3$
Boolean relation
that is not balanced.
Then, the $\kor{2}$ relation is cone-definable from $U$.
\end{theorem}

%\subsection{Classification}
%Summarize complete classification?

Combining the results in this section with the results in previous sections, allows us to give a  full classification of the sparsifiability of constraint languages that only contain relations of arity at most three. Observe that any $k$-ary relation $R$ such that $R \neq \emptyset$ and $\{0,1\}^k\setminus R \neq \emptyset$ cone-defines the $1$-\OR relation. Since we assume that $\Gamma$ is intractable in the next theorem, it follows that $k$ is always defined and $k\in \{1,2,3\}$.

\begin{theorem}
Let $\Gamma$ be an intractable Boolean constraint language such that each relation therein
has arity $\leq 3$. Let $k \in \mathbb{N}$ be the largest value for which $k$-\OR can be cone-defined from a relation in $\Gamma$. Then $\csp(\Gamma)$ has a kernel with $\Oh(n^k)$ constraints that can be encoded in $\Oh(n^k\log k)$ bits, but for any $\varepsilon > 0$ there is no kernel of size $\Oh(n^{k-\varepsilon})$, unless \containment.
%\begin{itemize}
%\item  If there exist $R \in \Gamma$ with $|R| = 7$, then $\csp(\Gamma)$ has a  kernel of size $\Oh(n^3)$ and admits no generalized kernel of size $\Oh(n^{3-\varepsilon})$ for any $\varepsilon > 0$.
%\item Else, if $\Gamma$ is not balanced, then $\csp(\Gamma)$ admits a kernel with $\Oh(n^{2})$ constraints of bitsize $\Oh(n^{2}\log n)$ and has no generalized kernel of size $\Oh(n^{2-\varepsilon})$ for any $\varepsilon > 0$.
%\item Otherwise,  $\csp(\Gamma)$ has a kernel with $\Oh(n)$ constraints and bitsize $\Oh(n\log n)$.
%\item Otherwise,
%
%\end{itemize}
\end{theorem}
\begin{proof}
To show that there is a kernel with $\Oh(n^k)$ constraints, we do a case distinction on $k$.
\begin{itemize}
\item $(k=1)$ If $k=1$, there is no relation in $\Gamma$ that cone-defines the $2$-\OR relation. It follows from Observation \ref{obs:Boolean-arity-one} and Theorems \ref{thm:Boolean-arity-two} and \ref{thm:not-balanced-implies-or} that thereby, $\Gamma$ is balanced. It now follows from Theorem \ref{thm:preserved_by_balanced} that $\csp(\Gamma)$ has a kernel with $\Oh(n)$ constraints that can be stored in $\Oh(n\log n)$ bits.
\item $(k=2)$ If $k=2$, there is no relation $R \in \Gamma$ with $|R| = 2^3-1=7$, as otherwise by Lemma \ref{lem:1-unsat-defines-OR} such a relation $R$ would cone-define $3$-\OR which is a contradiction. Thereby, it follows from Theorem \ref{thm:trivial-vs-nontrivial} that $\csp(\Gamma)$ has a sparsification with $\Oh(n^{3-1}) = \Oh(n^2)$ constraints that can be encoded in $\Oh(n^2\log n)$ bits.
\item $(k=3)$ Given an instance $(\C,V)$, it is easy to obtain a kernel of with $\Oh(n^3)$ constraints by simply removing duplicate constraints. This kernel can be stored in $\Oh(n^3)$ bits, by storing for each relation $R \in \Gamma$ and for each tuple $(x_1,x_2,x_3) \in V^3$ whether $R(x_1,x_2,x_3) \in \C$. Since $|\Gamma|$ is constant and there are $\Oh(n^3)$ such tuples, this results in using $\Oh(n^3)$ bits.
\end{itemize}

It remains to prove the lower bound.
By definition, there exists $R \in \Gamma$ such that $R$ cone-defines the $k$-\OR relation. Thereby, the result  follows immediately from Theorem \ref{thm:LB-by-defining-OR}. Thus,  $\csp(\Gamma)$ has no kernel of size $\Oh(n^{k-\varepsilon})$ for any $\varepsilon > 0$, unless \containment.
\end{proof}

\section{Conclusion}
In this paper we analyzed the best-case and worst-case sparsifiability of $\csp(\Gamma)$ for intractable finite Boolean constraint languages~$\Gamma$. First of all, we characterized those Boolean CSPs for which a nontrivial sparsification is possible, based on the number of non-satisfying assignments. Then we presented our key structural contribution: the notion of balanced constraint languages. We have shown that $\csp(\Gamma)$ allows a sparsification with $\Oh(n)$ constraints whenever $\Gamma$ is balanced. The constructive proof of this statement can be transformed into an effective algorithm to find a series of low-degree polynomials to capture the constraints, which earlier had to be done by hand. By combining the resulting upper and lower bound framework,  we fully classified the symmetric constraint languages for which $\csp(\Gamma)$ allows a linear sparsification. Furthermore, we fully classified the sparsifiability of $\csp(\Gamma)$ when $\Gamma$ contains relations of arity at most three, based on the arity of the largest \OR that can be cone-defined from~$\Gamma$. It follows from results of Lagerkvist and Wahlstr\"{o}m~\cite{LagerkvistW17} that for constraint languages of arbitrary arity, the exponent of the best sparsification size does not always match the arity of the largest \OR cone-definable from~$\Gamma$. (This will be described in more detail in the upcoming journal version of this work.) Hence the type of characterization we presented is inherently limited to low-arity constraint languages. It may be possible to extend our characterization to languages of arity at most four, however.

The ultimate goal of this line of research is to fully classify the sparsifiability of $\csp(\Gamma)$, depending on $\Gamma$. In particular, we would like to classify those $\Gamma$ for which $\Oh(n)$ sparsifiability is possible. In this paper, we have shown that $\Gamma$ being balanced is a sufficient condition to obtain a linear sparsification; it is tempting to conjecture that this condition is also necessary.

We conclude with a brief discussion on the relation between our polynomial-based framework for linear compression and the framework of Lagerkvist and Wahlstr\"{o}m~\cite{LagerkvistW17}. They used a different method for sparsification, based on embedding a Boolean constraint language~$\Gamma$ into a constraint language~$\Gamma'$ defined over a larger domain~$D$, such that~$\Gamma'$ is preserved by a Maltsev operation. This latter condition ensures that~CSP$(\Gamma')$ is polynomial-time solvable, which allows CSP$(\Gamma)$ to be sparsified to~$\Oh(n)$ constraints when~$D$ is finite. It turns out that the Maltsev-based linear sparsification is more general than the polynomial-based linear sparsification presented here: all finite Boolean constraint languages~$\Gamma$ that are balanced, admit a Maltsev embedding over a finite domain (the direct sum of the rings~$\mathbb{Z}/q_u\mathbb{Z}$ over which the capturing polynomials are defined) and can therefore be linearly sparsified using the algorithm of Lagerkvist and Wahlstr\"{o}m. Despite the fact that our polynomial-based framework is not more general than the Maltsev-based approach, it has two distinct advantages. First of all, there is a straight-forward decision procedure to determine whether a constraint can be captured by degree-1 polynomials, which follows from the proof of Theorem~\ref{thm:preserved_by_balanced}. To the best of our knowledge, no decision procedure is known to determine whether a Boolean constraint language admits a Maltsev embedding over a finite domain. The second advantage of our method is that when the polynomial framework for linear compression does not apply, this is witnessed by a relation in~$\Gamma$ that is violated by a balanced operation. As we have shown, in several scenarios this violation can be used to construct a sparsification lower bound to give provably optimal bounds.

It would be interesting to determine whether the Maltsev-based framework for sparsification is strictly more general than the polynomial-based framework. We are not aware of any concrete Boolean constraint language~$\Gamma$ for which CSP$(\Gamma)$ admits a Maltsev embedding over a finite domain, yet is not balanced.

\bibliography{refs}
%!TEX root = compression-ipec18.tex
%% this above is for sublime text.

\clearpage
\appendix
\section{Omitted proofs}
\label{app:missing-proofs}

\subsection{Proofs omitted from Section~\ref{sec:preliminaries}}

\begin{proof}[{Proof of Proposition \ref{prop:balanced-acts-as-alternating}}]
It suffices to show that if a relation $T$ is not balanced, then there exists an alternating operation
that does not preserve $T$.  Let $f$ be a $k$-ary balanced operation that does not preserve $T$.
Then there exist tuples $t^1, \ldots, t^k$ in $T$ such that $\alpha_1 t^1 + \cdots + \alpha_k t^k$
is not in $T$, where the sum of the $\alpha_i$ is equal to $1$ (and where we may assume that no $\alpha_i$ is equal to $0$).  For each positive $\alpha_i$,
replace $\alpha_i t^i$ in the sum with $t^i + \cdots + t^i$ ($\alpha_i$ times);
likewise, for each negative $\alpha_i$, replace $\alpha_i t^i$ in the sum with
$-t^i - \cdots -t^i$ ($-\alpha_i$ times).  Each tuple then has coefficient
$+1$ or $-1$ in the sum; since the sum of coefficients is $+1$,
by permuting the sum's terms, the coefficients can be made to alternate
between $+1$ and $-1$.
\end{proof}

For the proof of Proposition \ref{prop:lpt-from-cone-definability}, we will need  the following additional theorem.

\begin{theorem}
\label{thm:adding-constants}
(Follows from \cite{BulatovJeavonsKrokhin05-finitealgebras}.)
Let $\Gamma$ be a constraint language over a finite set $D$
such that each unary operation $u \colon D \to D$ that preserves $\Gamma$
is a bijection.
Then, there exists a linear-parameter transformation from $\csp(\Gamma^*)$ to $\csp(\Gamma)$.
\end{theorem}
Note that in particular, an \emph{intractable} Boolean constraint language can only be preserved by unary operations that are bijections. Hence for intractable Boolean~$\Gamma$, there is a linear-parameter transformation from~$\csp(\Gamma^*)$ to~$\csp(\Gamma)$.
\begin{proof}[Proof of Theorem~\ref{thm:adding-constants}]
The desired transformation is the final polynomial-time reduction given in the proof of Theorem 4.7
of~\cite{BulatovJeavonsKrokhin05-finitealgebras}.
This reduction translates an instance of $\csp(\Gamma^*)$ with $n$ variables
to an instance of $\csp(\Gamma \cup \{ =_D \})$ with $n + |D|$ variables;
here, $=_D$ denotes the equality relation on domain $D$.
Each constraint of the form $=_D(v,v')$ may be removed (while preserving satisfiability) by
taking one of the variables $v,v'$, and replacing each instance of that variable with the other.
The resulting instance of $\csp(\Gamma)$ has $\leq n + |D|$ variables.
\end{proof}

\begin{definition}
A relation $T \subseteq D^k$ is \emph{pp-definable} (short for \emph{primitive positive definable}) from a constraint language $\Gamma$ over $D$ if
there exists an instance $(\C,V)$ of $\csp(\Gamma)$
and there exist pairwise distinct variables $x_1, \ldots, x_k \in V$
such that, for each map $f \colon \{ x_1, \ldots, x_k \} \to \{ 0, 1 \}$,
it holds that $f$ can be extended to a satisfying assignment of the instance
if and only if $(f(x_1), \ldots, f(x_k)) \in T$.
\end{definition}

The following is a known fact;
for an exposition, we refer the reader to
Theorems 3.13 and 5.1 of \cite{Chen09-Rendezvous}.

\begin{proposition}\label{prop:gammastar-is-br}
If $\Gamma$ is an intractable Boolean constraint language, then
every Boolean relation is pp-definable from $\Gamma^*$.
\end{proposition}

\begin{proof}[Proof of Proposition \ref{prop:lpt-from-cone-definability}]
It suffices to give a linear-parameter transformation from
$\csp(\Gamma^* \cup \Delta)$ to $\csp(\Gamma^*)$, by Theorem~\ref{thm:adding-constants}.
Let $(\C,V)$ be an instance of $\csp(\Gamma^* \cup \Delta)$,
and let $n$ denote $|V|$.
We generate an instance $(\C',V')$ of $\csp(\Gamma^*)$ as follows.
\begin{itemize}

\item For each variable $v \in V$, introduce a primed variable $v'$.
By Proposition~\ref{prop:gammastar-is-br}, the relation $\neq$ (that is, the relation
$\{ (0,1),(1,0) \}$) is pp-definable from $\Gamma^*$.
Fix such a pp-definition, and let $d$ be the number of variables in the definition.
For each $v \in V$,
include in $\C'$ all constraints in the pp-definition of $\neq$,
but where the variables are renamed so that
$v$ and $v'$ are the distinguished variables, and the other variables are fresh.

The number of variables used so far in $\C'$ is $nd$.

\item For each $b \in \{ 0, 1 \}$, introduce a variable $z_b$,
and include the constraint $\{ (b) \}(z_b)$ in $\C'$.

\item For each constraint $T(v_1, \ldots, v_k)$ in $\C$ such that $T \in \Gamma^*$,
include the constraint in $\C'$.

\item For each constraint $T(v_1, \ldots, v_k)$ in $\C$ such that $T \in \Delta \setminus \Gamma^*$,
we use the assumption that $T$ is cone-definable from a relation in $\Gamma$ to include a constraint
in $\C'$ that has the same effect as $T(v_1, \ldots, v_k)$.
In particular, assume that $T$ is cone-definable from $U \in \Gamma$
via the tuple $(y_1, \ldots, y_{\ell})$,
and that $U$ has arity $\ell$.
Include in $\C'$ the constraint $U(w_1, \ldots, w_\ell)$,
where, for each $i \in [\ell]$, the entry $w_i$ is defined as follows:
\begin{equation*}
w_i = \begin{cases}
v_j \text{ if $y_i = x_j$,}\\
v'_j \text{ if $y_i = \neg x_j$,}\\
z_0 \text{ if $y_i = 0$, and}\\
z_1 \text{ if $y_i = 1$.}
\end{cases}
\end{equation*}

\end{itemize}
The set $V'$ of variables used in $\C'$ is the union of
$V \cup \{ v' ~|~ v \in V \} \cup \{ z_0,z_1 \}$ with the other variables used in the copies
of the pp-definition of $\neq$.
We have $|V'| = nd + 2$.
It is straightforward to verify that an assignment $f \colon V \to \{ 0, 1 \}$
satisfies $\C$ if and only if there exists an assignment $f' \colon V' \to \{ 0, 1 \}$ of $f$
that satisfies $\C'$.
\end{proof}

\subsection{Proofs omitted from Section \ref{sec:trivial-vs-nontrivial}}
\label{app:trivial-vs-nontrivial}
We start by proving the main Theorem of this section, using the other lemmas in the section.
\begin{proof}[Proof of Theorem \ref{thm:trivial-vs-nontrivial}]
Suppose that for all $R \in \Gamma$, it holds that $|R| \neq 2^k-1$. We give the following kernelization procedure.  Suppose we are given an instance of $\csp(\Gamma)$, with set of constraints $\C$. We show how to define $\C' \subseteq \C$. For each constraint $R(x_1,\ldots,x_\ell) \in \C$ where $R$ is a relation of arity $\ell < k$, add one such constraint to $\C'$ (thus removing duplicate constraints). Note that this adds at most $\Oh(n^\ell)$ constraints for each $\ell$-ary relation $R \in \Gamma$.

For a $k$-ary relation $R\in\Gamma$, let $\C_R$ contain all constraints of the form $R(x_1,\ldots,x_k)$. For all $k$-ary relations $R$ with $|R| < 2^k-1$, apply Lemma \ref{lem:2-unsat-implies-kernel} to obtain $\C'_R \subseteq \C_R$ such that $|\C'_R| = \Oh(n^{k-1})$ and any Boolean assignment satisfying $\C'_R$ also satisfies $\C_R$. Add $\C'_R$ to $\C'$. This concludes the definition of $\C'$. Note that the procedure removes constraints of the form $R(x_1,\ldots,x_k)$ with $|R|=2^k$, as these are always satisfied. It is easy to verify that $|\C'| \leq |\Gamma|\cdot \Oh(n^{k-1}) = \Oh(n^{k-1})$. Since each constraint can be stored in $\Oh(\log n)$ bits, this gives a kernel of bitsize $\Oh(n^{k-1}\log n)$.

Suppose that there exists $R \in \Gamma$ with $|R| = 2^k-1$. It follows from Lemma \ref{lem:1-unsat-defines-OR} that $R$ cone-defines $k$-\OR. Since $\Gamma$ is intractable, it now follows from Theorem \ref{thm:LB-by-defining-OR} that $\csp(\Gamma)$ has no generalized kernel of size $\Oh(n^{k-\varepsilon})$, unless \containment.
\end{proof}

To prove Theorem \ref{thm:only-need-polynomials}, we  use the following two theorems, that were proven by a subset of the current authors \cite{JansenP18}. We recall the required terminology. Let $E$ be a ring. Define \rootCSP over $E$ as the problem whose input consists of a set $L$ of polynomial equalities over $E$ of degree at most $d$, over a set of variables $V$. Each equality is of the form $p(x_1,\ldots,x_k) = 0$ (over $E$). The question is whether there exists a Boolean assignment to the variables in $V$ that satisfies all equalities in $L$.

\begin{theorem}[{\cite[Theorem 16]{JansenP18}}]\label{thm:subset_kernel_modm}
There is a polynomial-time algorithm that, given an instance~$(L,V)$ of \rootCSP over $\mathbb{Z}/m\mathbb{Z}$ for some fixed integer $m \geq 2$ with~$r$ distinct prime divisors, outputs an equivalent instance~$(L',V)$ of \rootCSP over $\mathbb{Z}/m\mathbb{Z}$ with at most $r\cdot(n^d+1)$ constraints such that~$L' \subseteq L$.
\end{theorem}

Note that if $m$ is a prime power, $m$ has only one distinct prime divisor and thereby $r=1$ in the above theorem statement.

We say a field $F$ is \emph{efficient} if the field operations and Gaussian elimination can be done in polynomial time in the size of a reasonable input encoding.

\begin{theorem}[{\cite[Theorem 5]{JansenP18}}]\label{thm:subset_kernel_Q}
There is a polynomial-time algorithm that, given an instance $(L,V)$ of \rootCSP over an efficient field $F$, outputs an equivalent instance~$(L',V)$ with at most $n^d + 1$ constraints such that $L' \subseteq L$.
\end{theorem}

Observe that the above theorem statement in particular applies to instances of \rootCSP over $\mathbb{Q}$, since $\mathbb{Q}$ is an efficient field.

\begin{proof}[Proof of Theorem \ref{thm:only-need-polynomials}]
Let $\C$ be a set of constraints over $R$ and let $V$ be the set of variables used. We will create $| \{0,1\}^k\setminus R|$ instances of \rootCSP with variable set $V$. For each $u \in R \setminus \{0,1\}^k$, we create an instance $(L_u,V)$ of \rootCSP over $E_u$, as follows. Choose a ring~$E_u \in \{\mathbb{Q}\}	\cup \{\mathbb{Z}/q_u\mathbb{Z} \mid q_u \text{ is a prime power}\}$ and a polynomial $p_u$ over $E_u$  such that \eqref{eq:sat-R} and \eqref{eq:unsat-u} are satisfied for $u$. For each constraint $(x_1,\ldots,x_k) \in \C$, add the equality $p_u(x_1,\ldots,x_k)= 0$ to the set $L_u$; note that these are equations over the ring $E_u$. Let $L:= \bigcup_{u\notin R} L_u$ be the union of all created sets of equalities. From this construction, we obtain the following claim.

\begin{claim}\label{claim:poly-implies-constraints}
Any Boolean assignment $f$ that satisfies all equalities in $L$, satisfies all constraints in $\C$.
\end{claim}
\begin{claimproof}
Let $f$ be a Boolean assignment that satisfies all equalities in $L$. Suppose $f$ does not satisfy all equalities in $\C$, thus there exists $(x_1,\ldots,x_k) \in \C$, such that $(f(x_1),\ldots,f(x_k)) \notin R$. Let $u := (f(x_1),\ldots,f(x_k))$. Since $u \notin R$, the equation $p_u(x_1,\ldots,x_k) = 0$ was added to $L_u\subseteq L$. However, it follows from \eqref{eq:unsat-u} that $p_u(f(x_1),\ldots,f(x_k))\neq 0$, which contradicts the assumption that $f$ satisfies all equalities in $L$.
%
%Let $f$ be a Boolean assignment that satisfies all constraints in $\C$. Let $p_u(x_1,\ldots,x_k) \equiv_{q_u} 0$ be an equality in $L$. Thereby, $(x_1,\ldots,x_k) \in \C$, and since $f$ is a satisfying assignment it follows that $(f(x_1),\ldots,f(x_k)) \in R$. It follows from \eqref{eq:sat-R} that $p_u(f(x_1),\ldots,f(x_k))\equiv_{q_u} 0$.
\end{claimproof}
For each instance $(L_u,V)$ of \rootCSP over $E_u$ with $E_u \neq \mathbb{Q}$, apply Theorem \ref{thm:subset_kernel_modm} to obtain an equivalent instance $(L_u',V)$ with $L_u'\subseteq L_u$ and $|L_u'| =\Oh(n^{d})$. Similarly, for each instance $(L_u,V)$ of \rootCSP over $E_u$ with $E_u = \mathbb{Q}$, apply Theorem \ref{thm:subset_kernel_Q} and obtain an equivalent instance $(L_u',V)$ with $L_u'\subseteq L_u$ and $|L_u'| =\Oh(n^{d})$. Let $L' := \bigcup L_u'$.
By this definition, any Boolean assignment satisfies all equalities in $L$, if and only if it satisfies all equalities in $L'$. Construct $\C'$ as follows. For any $(x_1,\ldots,x_k) \in \C$, add $(x_1,\ldots,x_k)$ to $\C'$ if there exists $u \in \{0,1\}^k\setminus R$ such that $p_u(x_1,\ldots,x_k) = 0 \in L'$. Hereby, $\C' \subseteq \C$. The following two claims show the correctness of this sparsification procedure.
\begin{claim}\label{claim:kernel-correct}
Any Boolean assignment $f$ satisfies all constraints in $\C'$, if and only if it satisfies all constraints in \C.
\end{claim}
\begin{claimproof}
Since $\C'\subseteq \C$, it follows immediately that any Boolean assignment satisfying the constraints in $\C$ also satisfies all constraints in $\C'$. It remains to prove the opposite direction.

Let $f$ be a Boolean assignment satisfying all constraints in $\C'$. We show that $f$ satisfies all equalities in $L'$. Let $p_u(x_1,\ldots,x_k) = 0 \in L'$. Thereby, $(x_1,\ldots,x_k) \in \C'$ and since $f$ is a satisfying assignment, $(f(x_1),\ldots,f(x_k)) \in R$.  It follows from property \eqref{eq:sat-R} that $p_u(f(x_1),\ldots,f(x_k)) = 0 $ as desired.

Since $f$ satisfies all equalities in $L'$, it satisfies all equalities in $L$ by the choice of $L'$. It follows from Claim \ref{claim:poly-implies-constraints} that thereby $f$ satisfies all constraints in $\C$.
\end{claimproof}

\begin{claim}\label{claim:kernel-size}
$|\C'| = \Oh(n^d)$.
\end{claim}
\begin{claimproof}
By the construction of $\C'$, it follows that $|\C'| \leq |L'|$. We know $|L'| = \sum_{u \notin R} |L_u'|\leq \sum_{u \notin R} \Oh(n^d) \leq  2^k \Oh(n^d) = \Oh(n^d)$, as $k$ is considered constant.
\end{claimproof}
Claims \ref{claim:kernel-correct} and \ref{claim:kernel-size} complete the proof of Theorem~\ref{thm:only-need-polynomials}.
\end{proof}

Next, we present the full construction of degree-$(k-1)$ polynomials that capture relations~$R \subseteq \{0,1\}^k$ for which~$|R| < 2^k - 1$.

\begin{proof}[Proof of Lemma \ref{lem:2-unsat-implies-kernel}]
We will prove this by showing that for every $u \in \{0,1\}^k \setminus R$, there exists a $k$-ary polynomial $p_u$ over $\mathbb{Q}$ of degree at most $k-1$ satisfying
\eqref{eq:sat-R} and \eqref{eq:unsat-u}, such that the result follows  from Theorem~\ref{thm:only-need-polynomials}.

We will prove the existence of such a polynomial by induction on $k$. For $k=1$, the lemma statement implies that $R  = \emptyset$. Thereby, for any $u \notin R$, we simply choose $p_u(x_1) := 1$. This polynomial satisfies the requirements, and has degree $0$.

Let $k > 1$ and let $u=(u_1,\ldots,u_k) \in \{0,1\}^k \setminus R$. Since $|R| < 2^k-1$, we can choose $w = (w_1,\ldots,w_k)$ such that $w \in \{0,1\}^k \setminus R$ and $w \neq {u}$. Choose such $w$ arbitrarily, we now do a case distinction.

\textbf{(There exists no $i \in [k]$ for which $u_i = w_i$)} This implies $u_i = \neg w_i$ for all $i$. One may note that for $u = (0,\ldots,0)$ and $w = (1,\ldots,1)$ this situation corresponds  to \textsc{monotone $k$-nae-sat}. We show that there exists a polynomial $p_u$ such that $p_u(u_1,\ldots,u_k) \neq 0$,  and $p_u(x_1,\ldots,x_k) = 0$ for all $(x_1,\ldots,x_k) \in R$. Hereby $p_u$ satisfies conditions \eqref{eq:sat-R} and \eqref{eq:unsat-u} for  $u$. For $i \in [k]$, define $r_i(x) := (1-x)$ if $u_i = 1$ and $r_i(x) := x$ if $u_i = 0$. It follows immediately from this definition that $r_i(u_i) = 0$ and~$r_i(w_i) = 1$ for all $i \in [k]$. Define
\[p_u(x_1,\ldots,x_k) := \prod_{i=1}^{k-1} \left(i - \sum_{j=1}^k r_j(x_j)\right). \]
By this definition, $p_u$ has degree $k-1$. It remains to verify that $p_u$ has the desired properties. First of all, since $\sum_{j=1}^k r_j(u_j) = 0$ by definition, it follows that
\[p_u(u_1,\ldots,u_k) = \prod_{i=1}^{k-1} i \neq 0,\]
as desired. Since~$r_i(w_i) = 1$ for all~$i$, we obtain $p_u(w_1,\ldots,w_k) = \prod_{i=1}^{k-1} (i-k) \neq 0$, which is allowed since $w \notin R$.
It is easy to verify that in all other cases, $\sum_{j=1}^k r_j(x_j) \in \{1,2,\ldots,k-1\}$ and thereby one of the terms of the product is zero, implying $p_u(x_1,\ldots,x_k) = 0$.

\textbf{(There exists $i \in [k]$, such that $u_i = w_i$)} Let ${u'}$ and ${w'}$ be defined as the results of removing coordinate $i$ from ${u}$ and ${w}$ respectively. Note that ${u'} \neq {w'}$. Define
\[R' := \{(x_1,\ldots,x_{i-1},x_{i+1},\ldots,x_k) \mid (x_1,\ldots,x_{i-1},u_i,x_{i+1}, \ldots, x_k)\in R\}.\]
By this definition, ${u'},{w'} \notin R'$ and thereby $R'$ is a $(k-1)$-ary relation with $|R'| < 2^{k-1}-1$. By the induction hypothesis, there exists a polynomial $p_{u'}$ of degree at most $k-1$, such that $p_{u'}(u_1',\ldots,u_{k-1}') \neq 0$ and $p_{u'}(x_1',\ldots,x_{k-1}') = 0$ for all ${x'} \in R'$. Now define
\[p_u(x_1,\ldots,x_k) := (1-x_i-u_i)\cdot p_{u'}(x_1,\ldots,x_{i-1},x_{i+1},\ldots,x_k).\]
We show that $p_u$ has the desired properties. By definition, $p_u$ has the degree of $p_{u'}$ plus one. Since $p_{u'}$ has degree $k-2$ by the induction hypothesis, it follows that $p_u$ has degree $k-1$. Let $(x_1,\ldots,x_k) \in R$. We do a case distinction on the value taken by $x_i$.
\begin{itemize}
\item $x_i \neq u_i$. In this case, $(1-x_i-u_i) = 0$, and thereby $p_u(x_1,\ldots,x_k) = 0$, thus satisfying condition \eqref{eq:sat-R}.
\item $x_i = u_i$. Since ${x} = (x_1,\ldots,x_{i-1},u_i,x_{i+1}, \ldots,x_k) \in R$, it follows that $(x_1,\ldots,x_{i-1}, x_{i+1},\allowbreak\ldots,x_k) \in R'$. By definition of $p_{u'}$, it follows that $p_{u'}(x_1,\ldots,x_{i-1},x_{i+1},\ldots,x_k) = 0$ and thus $p_u(x_1,\ldots,x_k) = 0$, showing \eqref{eq:sat-R}.
\end{itemize}
It remains to show that $p_u(u_1,\ldots,u_k) \neq 0$. This follows from $(1-u_i-u_i) \in \{-1,1\}$, and $p_{u'}(u_1,\ldots,u_{i-1},u_{i+1},\ldots,u_k) \neq 0$, showing that \eqref{eq:unsat-u} holds.

 Since we have shown for all $u \in \{0,1\}^k \setminus R$ that there exists a polynomial $p_u$ over $\mathbb{Q}$ satisfying \eqref{eq:sat-R} and \eqref{eq:unsat-u}, the proof of Lemma~\ref{lem:2-unsat-implies-kernel} now follows from Theorem \ref{thm:only-need-polynomials}.
\end{proof}

The following theorem presents several lower bounds, which combine various existing results from the literature.

\begin{proof}[Proof of Theorem \ref{thm:LB-by-defining-OR}]
We do a case distinction on $k$.

$(k=1)$ Suppose that there exists~$\varepsilon > 0$ such that $\csp(\Gamma)$ has a (generalized) kernel of size~$\Oh(n^{1-\varepsilon})$. Using this hypothetical generalized kernel, one could obtain a polynomial-time algorithm that takes as input a series of instances~$(\mathcal{C}_1, V_1), \ldots, (\mathcal{C}_t, V_t)$ of $\csp(\Gamma)$, and outputs in polynomial time an instance~$x^*$ of some fixed decision problem~$L$ such that:
\begin{itemize}
	\item $x^* \in L$ if and only if \emph{all}~$(\mathcal{C}_i, V_i)$ are \yes-instances of $\csp(\Gamma)$, and
	\item $x^*$ has bitsize~$\Oh(N^{1-\varepsilon})$, where~$N := \sum_{i=1}^t |V_i|$.
\end{itemize}
To obtain such an \textsc{and}-compression algorithm from a hypothetical generalized kernel of $\csp(\Gamma)$ into a decision problem~$L$, it suffices to do the following:
\begin{enumerate}
	\item On input a series of instances~$(\mathcal{C}_1, V_1), \ldots, (\mathcal{C}_t, V_t)$ of~$\csp(\Gamma)$, form a new instance~$(\mathcal{C}^* := \bigcup _{i=1}^t \mathcal{C}_i, V^* := \bigcup _{i=1}^t V_i)$ of~$\csp(\Gamma)$. Hence we take the disjoint union of the sets of variables and the sets of constraints, and it follows that the new instance has answer \yes if and only if all the inputs~$(\mathcal{C}_i, V_i)$ have answer \yes.
	\item Run the hypothetical generalized kernel on~$(\mathcal{C}^*, V^*)$, which has~$|V^*| = N$ variables and is therefore reduced to an equivalent instance~$x^*$ of~$L$ with bitsize~$\Oh(N^{1-\varepsilon})$.
\end{enumerate}
If we apply this \textsc{and}-compression scheme to a sequence of~$t_1(m) := m^\alpha$ instances of~$m$ bits each (which therefore have at most~$m$ variables each), the resulting output has~$\Oh(|V^*|^{1- \varepsilon}) = \Oh((m \cdot m^\alpha)^{1-\varepsilon}) = \Oh(m^{(1+\alpha)(1-\varepsilon)})$ bits.
By picking~$\alpha$ large enough that it satisfies~$(1+\alpha)(1-\varepsilon) \leq \alpha$, we therefore compress a sequence of~$t_1(m)$ instances of bitsize~$m$ into one instance expressing the logical AND, of size at most~$t_2(m) \leq \Oh(m^{(1+\alpha)(1-\varepsilon)}) \leq C \cdot t_1(m)$ for some suitable constant~$C$. Drucker~\cite[Theorem 5.4]{Drucker15} has shown that an error-free deterministic \textsc{and}-compression algorithm with these parameters for an NP-complete problem into a fixed decision problem~$L$, implies \containment. Hence the lower bound for~$k=1$ follows since~$\csp(\Gamma)$ is NP-complete.

($k \geq 2)$ For~$k \geq 2$, we prove the lower bound using a linear-parameter transformation (recall Definition \ref{def:lpt}). Let $\Delta$ be the set of $k$-ary relations given by $\Delta := \{\{0,1\}^k\setminus \{u\} \mid u \in \{0,1\}^k\}$. In particular, note that $\Delta$ contains the \kor{k} relation. Since $R$ cone-defines \kor{k}, it is easy to see that by variable negations, $R$ cone-defines all relations in $\Delta$. Thereby, it follows from Proposition \ref{prop:lpt-from-cone-definability} that there is a linear-parameter transformation from $\csp(\Gamma^*\cup \Delta)$ to $\csp(\Gamma)$. Thus, to prove the lower bound for $\csp(\Gamma)$, it suffices to prove the desired lower bound for $\csp(\Gamma^*\cup \Delta)$.

$(k=2)$ If $k=2$, we do a linear-parameter transformation from \textsc{Vertex Cover} to $\csp(\Gamma^*\cup \Delta)$. Since it is known that \textsc{Vertex Cover} parameterized by the number of vertices $n$ has no generalized kernel of size $\Oh(n^{2-\varepsilon})$ for any $\varepsilon > 0$, unless \containment~\cite{DellM14}, the result will follow.

Suppose we are given a graph $G = (V,E)$ on~$n$ vertices and integer $k\leq n$, forming an instance of the \textsc{Vertex Cover} problem. The question is whether there is a set~$S$ of~$k$ vertices, such that each edge has at least one endpoint in~$S$. We create an equivalent instance $(\mathcal{C},V')$ of $\csp(\Gamma^*)$ as follows. We introduce a new variable $x_{v}$ for each $v \in V$. For each edge $\{u,v\} \in E$, we add the constraint $\kor{2}(x_u,x_v)$ to $\C$.

At this point, any vertex cover in $G$ corresponds to a satisfying assignment, and vice versa. It remains to ensure that the size of the vertex cover is bounded by $k$. Let $H_{n,k}$ be the $n$-ary relation given by $H_{n,k} = \{(x_1,\ldots,x_n) \mid x_i \in \{0,1\} \text{ for all } i\in[n] \text{ and } \sum_{i\in[n]}x_i = k\}$. By Proposition \ref{prop:gammastar-is-br}, we obtain that $\Gamma^*$ pp-defines all Boolean relations. It follows from \cite[Lemma 17]{LagerkvistW17} that $\Gamma^*$ pp-defines $H_{n,k}$ using $\Oh(n+k)$ constraints and $\Oh(n+k)$ existentially quantified variables. We add the constraints from this pp-definition to \C, and add the existentially quantification variables to~$V'$. This concludes the construction of \C. It is easy to see that \C has a satisfying assignment if and only if $G$ has a vertex cover of size $k$. Furthermore, we used $\Oh(n+k) \in \Oh(n)$ variables and thereby this is a linear-parameter transformation from \textsc{Vertex Cover} to $\csp(\Gamma^*\cup \Delta)$.

$(k\geq3)$
In this case there is a trivial linear-parameter transformation from $\csp(\Delta)$ to $\csp(\Gamma^* \cup \Delta)$. It is easy to verify that $\csp(\Delta)$ is equivalent to $k$-CNF-SAT. The result now follows from the fact that for $k \geq 3$, $k$-CNF-SAT has no kernel of size $\Oh(n^{k-\varepsilon})$ for any $\varepsilon > 0$, unless \containment~\cite{DellM14}.
\end{proof}

As the last result of the section, we prove that $k$-ary Boolean relations with exactly one falsifying assignment cone-define \kor{k}.

\begin{proof}[Proof of Lemma \ref{lem:1-unsat-defines-OR}]
Let ${u} = (u_1,\ldots,u_k)$ be the unique $k$-tuple not contained in $R$. Define the tuple $(y_1,\ldots,y_k)$ as follows. Let $y_i := x_i$ if $u_i = 0$, and let $y_i := \neg x_i$ otherwise. Clearly, this satisfies the first two conditions of cone-definability. It remains to prove the last condition. Let $f \colon\{x_1,\ldots,x_m\} \rightarrow \{0,1\}$. Suppose $(f(x_1),\ldots,f(x_k)) \in k\text{-\OR}$. We show $(\hat{f}(y_1),\ldots,\hat{f}(y_k)) \in R$. Since $(f(x_1),\ldots,f(x_k)) \in k\text{-\OR}$, there exists  at least one $i \in [k]$ such that $f(x_i) \neq 0$. Thereby, $\hat{f}(y_i) \neq u_i$ and thus $(\hat{f}(y_1),\ldots,\hat{f}(y_k)) \neq u$, implying $(\hat{f}(y_1),\ldots,\hat{f}(y_k)) \in R$.

Suppose $(f(x_1),\ldots,f(x_k)) \notin k\text{-\OR}$, implying $f(x_i) = 0$ for all $i \in [k]$. But this implies $\hat{f}(y_i) = u_i$ for all $i \in [k]$ and thus $(\hat{f}(y_1),\ldots,\hat{f}(y_k)) = {u} \notin R$.
\end{proof}

\subsection{Proofs omitted from Section \ref{sec:balanced-to-linear}}
\label{app:balanced-to-linear}
To give the proofs that were omitted from Section \ref{sec:balanced-to-linear}, we need the following additional definitions.
\begin{definition}
We say an $m\times n$ matrix $A$ is a \emph{diagonal matrix}, if all entries $a_{i,j}$
 with $i \neq j$ are zero. Thus, all non-zero elements occur on the diagonal.
\end{definition}
Note that by the above definition of diagonal matrices,  a matrix can be diagonal even if it is not a square matrix.

We denote the greatest common divisor of two integers $x$ and $y$ as $\text{gcd}(x,y)$. Recall that by B\'{e}zout's lemma, if $\text{gcd}(x,y) = z$ then there exist integers $a$ and $b$ such that $ax + by = z$.  We will use $x \mid y$ to indicate that $x$ divides $y$ (over the integers) and $x \nmid y$ to indicate that it does not. The proof of the following lemma was contributed by Emil Je\v{r}\'abek.

\begin{proof}[Proof of Lemma \ref{lem:prime-powers-to-int}]
We prove the contrapositive. Suppose $u \notin \text{span}_{\mathbb{Z}}(\{s_1,\ldots,s_m\})$, thus $u$ cannot be written as a linear combination of the rows of $S$ over $\mathbb{Z}$; equivalently, the system $yS = u$ has no solutions for $y$ over $\mathbb{Z}$. We will show that there exists a prime power~$q$, such that $yS \equiv_q u$ has no solutions over $\mathbb{Z}/q\mathbb{Z}$ and thus $u \notin \text{span}_q(\{s_1,\ldots,s_m\})$.

There exist an $m\times m$ matrix $M$ and an $n\times n$ matrix $N$
 over $\mathbb{Z}$, such that $M$ and $N$ are invertible over $\mathbb{Z}$ and furthermore $S' := MSN$ is in Smith Normal Form (cf.~\cite[Theorem 368]{gockenbach2011finite}).  In particular, this implies that $S'$ is a diagonal matrix.
Define $u' := uN$.
\begin{claim}\label{claim:SNF-equivalent}
If~$y'S'=u'$ is solvable for $y'$ over $\mathbb{Z}$, then $yS = u$ is solvable for $y$ over~$\mathbb{Z}$.
\end{claim}
\begin{claimproof}
Consider $y'$ such that $y'S'=u'$. One can verify that $y := y'M$ solves $yS = u$, as
\[yS = y'M S= y'MSNN^{-1} = y'S'N^{-1} =u'N^{-1} = uNN^{-1} = u.\qedhere\]
\end{claimproof}

\begin{claim} \label{claim:modsolutions:to:normal}
Let~$q \in \mathbb{N}$. If~$yS \equiv_q u$ is solvable for~$y'$, then~$y'S' \equiv_q u'$ is solvable for~$y$.
\end{claim}
\begin{claimproof}
Let $y$ be such that $yS \equiv_q u$. Define $y' := yM^{-1}$. We verify that $y'S' \equiv_q u'$ as follows.
\[y'S' = y'MSN = yM^{-1}MSN = ySN \equiv_q uN = u'.\qedhere\]
\end{claimproof}

Using these two claims, our proof by contraposition proceeds as follows. From our starting assumption~$u \notin \text{span}_{\mathbb{Z}}(\{s_1,\ldots,s_m\})$, it follows by Claim~\ref{claim:SNF-equivalent} that~$y'S'=u'$ has no solution~$y'$ over~$\mathbb{Z}$. Below, we prove that this implies there exists a prime power $q$ such that $y'S' \equiv_q u'$ is unsolvable. By Claim~\ref{claim:modsolutions:to:normal} this will imply that~$y S \equiv_q u$ is unsolvable and complete the proof. 

Suppose $y'S' = u'$ has no solutions over~$\mathbb{Z}$. Since all non-zero elements of $S'$ are on the diagonal, this implies that either there exists $i \in [n]$, such that $u'_i$ is not divisible by $s'_{i,i}$, or $s'_{i,i}$ is zero while $u'_i \neq 0$. We finish the proof by a case distinction.

\begin{itemize}
\item Suppose there exists $i \in [n]$ such that $s'_{i,i} = 0$, while $u'_i \neq 0$. Choose a prime power $q$ such that $q \nmid u'_i$. It is easy to see that thereby, $u'_i \not\equiv_q 0$. Since $s'_{i,i} \equiv_q 0$ holds trivially in this case, the system has $y'S' \equiv_q u'$ no solution.
\item Otherwise, there exists $i \in [n]$ such that $s'_{i,i} \nmid u'_i$. Choose a prime power $q$ such that $q \nmid u'_i$ and $q \mid s'_{i,i}$. Such a prime power can be chosen by letting $q := p^\ell$ for a prime $p$ that occurs $\ell \geq 1$ times in the prime factorization of $s'_{i,i}$, but less often in the prime factorization of $u'_i$. Thereby, $u'_i \not\equiv_q 0$, while $s'_{i,i} \equiv_q 0$. It again follows that the system $y'S' \equiv_q u'$ has no solutions.
\qedhere
\end{itemize}
\end{proof}

We remark that the proof of Lemma~\ref{lem:prime-powers-to-int} can be made constructive in the following sense: there is an algorithm that either finds a linear combination showing that~$u \in \text{span}_{\mathbb{Z}}(\{s_1,\ldots,s_m\})$, or produces a prime power~$q$ for which~$u \notin \text{span}_q(\{s_1,\ldots,s_m\})$. The running time of this algorithm is superpolynomial due to the necessity to factor integers, but for moderately-sized integers this is not a big issue in practice.

\begin{proof}[Proof of Lemma \ref{lem:no-solution-implies-span}]
Let $A'$ be the $(m-1) \times n$ matrix consisting of the first $m-1$ rows of $A$. Find the Smith normal form \cite{gockenbach2011finite} of $A'$ over $\mathbb{Z}$, thus there exist an $(m-1)\times (m-1)$ matrix $M'$ and an $n \times n$ matrix $N$, such that $S':= M'A'N$ is in Smith Normal Form and $M'$ and $N$ are invertible over $\mathbb{Z}$. (The only property of Smith Normal Form we rely on is that~$S'$ is a diagonal matrix.)

We show that similar properties hold over $\mathbb{Z}/q\mathbb{Z}$. Let $(M')^{-1}$, $N^{-1}$ be the inverses of $M'$ and $N$ over $\mathbb{Z}$. It is easy to verify that $NN^{-1} = I \equiv_q I$ and $M'(M')^{-1} = I \equiv_q I$, such that $M'$ and $N$ are still invertible over $\mathbb{Z}/q\mathbb{Z}$. Furthermore, $S'\ (\text{mod }{q})$ remains a diagonal matrix.

Define $M$ to be the following $m\times m$ matrix
\[M :=
\left(
\begin{array}{c|c}
M' & \vect{0} \\
\hline
\vect{0} & 1
\end{array}
\right),\]
then $M$ has an inverse over $\mathbb{Z}/q\mathbb{Z}$ that is given by the following matrix
 \[M^{-1} \equiv_q
\left(
\begin{array}{c|c}
(M')^{-1} & \vect{0} \\
\hline
\vect{0} & 1
\end{array}
\right).\]
Define $S:= MAN$ and verify that
\begin{equation}\label{eq:definition_of_S}
S:= MAN \equiv_q \left(
\begin{array}{c}
S' \\
\hline
a_mN
\end{array}
\right),
\end{equation}
meaning that the first $m-1$ rows of $S$ are equal to the first $m-1$ rows of  $S'$, and the last row of $S$ is given by the row vector $a_mN$.

The following two claims will be used to show that proving the lemma statement for matrix $S$, will give the desired result for $A$.
\begin{claim}\label{claim:solutions-equivalent}
Let $b:= (0,\ldots,0,c)$ for some constant $c$. The system $Sx' \equiv_q b$ has a solution, if and only if the system $Ax\equiv_q b$ has a solution.
\end{claim}
\begin{claimproof}
Let $x$ be a solution for $Ax\equiv_q b$. Define $x' := N^{-1}x$. Then $MANx' \equiv_q MAx \equiv_q Mb$. Observe that by the definitions of $M$ and $b$, $Mb \equiv_q b$, which concludes this direction of the proof.

For the other direction, let $x'$ be a solution for $MANx' \equiv_q b$. Define $x := Nx'$. Then $M^{-1}MANx' \equiv_q M^{-1}b$ and thus $ANx' \equiv_q M^{-1}b$ and thereby $Ax \equiv_q M^{-1}b$. By the definition of $M^{-1}$ and $b$, we again have $M^{-1}b \equiv_q b$.
\end{claimproof}

\begin{claim}\label{claim:span-equivalent}
$s_m \in \text{span}_q(\{s_1,\ldots,s_{m-1}\})$ if and only if $a_m \in \text{span}_q(\{a_1,\ldots,a_{m-1}\})$.
\end{claim}
\begin{claimproof}
Suppose $s_m \in \text{span}_q(\{s_1,\ldots,s_{m-1}\})$. This implies that there exist $\alpha_1,\ldots,\alpha_{m-1}$ such that $\sum_{i\in[m-1]} \alpha_is_i \equiv_q s_m \equiv_q a_mN$. Thus,
$\sum_{i\in[m-1]} \alpha_is'_i \equiv_q a_mN$, and for $\alpha = (\alpha_1,\ldots,\alpha_{m-1})$ we therefore have
$\alpha S' \equiv_q a_mN$, implying
$(\alpha M')A'N \equiv_q a_mN$. Since $N$ is invertible, it follows that
\[(\alpha M')A' \equiv_q a_m\] and thus $a_m \in \text{span}_q(\{a_1,\ldots,a_{m-1}\})$.

For the other direction, suppose $a_m \in \text{span}_q(\{a_1,\ldots,a_{m-1}\})$. Thus, there exists $\alpha \equiv_q (\alpha_1,\ldots,\alpha_{m-1})$ such that $\alpha A' \equiv_qa_m$. Let $\alpha' := \alpha (M')^{-1}$. Then \[\alpha'S' \equiv_q \alpha'M'A'N \equiv_q \alpha A'N \equiv_q a_mN \equiv_q s_m,\] and now it follows from the definition of $S$ given in \eqref{eq:definition_of_S} that $s_m \in \text{span}_q(\{s_1,\ldots,s_{m-1}\})$.
\end{claimproof}

It follows from Claims \ref{claim:solutions-equivalent} and \ref{claim:span-equivalent}, that it suffices to show that if $Sx = (0,\ldots,0,c)^T$ has no solutions for any $c \not\equiv_q 0$, then $s_m \in \text{span}_q(\{s_1,\ldots,s_{m-1}\})$. So suppose $s_m \notin \text{span}_q(\{s_1,\ldots,s_{m-1}\})$, we show that the system has a solution for some non-zero $c$. Observe that since $S'$ (the first $m-1$ rows of $S$) is a diagonal matrix, there must exist $i\in [m-1]$ for which  there is no $\alpha_i$ satisfying $s_{i,i} \cdot \alpha_i \equiv_q s_{m,i}$. Otherwise, it is easy to see that $\sum_{i\in[m-1]}\alpha_is_i \equiv_q s_m$, contradicting that $s_m \notin \text{span}_q(\{s_1,\ldots,s_{m-1}\})$. We now do a case distinction.

Suppose there exists $i \in [m-1]$ such that $s_{i,i} \equiv_q 0$, while $s_{m,i}\not\equiv_q 0$. Let $x = (0,\ldots,0,1,0,\ldots,0)$ be the vector with $1$ in the $i$'th position and zeros in all other positions. It is easy to verify that $Sx \equiv_q (0,\ldots,0,s_{m,i})^T$ and thereby the system $Sx \equiv_q b$ has a solution for $c=s_{m,i}$.

Otherwise, choose $i$ such that there exists no integer $\alpha_i$ satisfying $s_{i,i}\cdot \alpha_i \equiv_q s_{m,i}$ and $s_{i,i}\not\equiv_q 0$. It is given that $q$ is a prime power, let $q = p^k$ for prime $p$. Let $0\leq \ell < k$ be the largest integer such that $p^\ell \mid s_{i,i}$ over the integers. We consider the following two cases.
\begin{itemize}
\item Suppose $p^\ell \mid s_{m,i}$. Let $c$ such that $s_{i,i} \equiv_q c \cdot p^\ell$ and choose $d$ such that $s_{m,i}\equiv_q d \cdot p^\ell$. Note that $\text{gcd}(c,q) = 1$. It follows from B\'{e}zout's lemma that $c$ has an inverse $c^{-1}$ such that $cc^{-1} \equiv_q 1$. Then
    \[s_{i,m} \equiv_q (d\cdot c^{-1})s_{i,i}, \]
    which is a contradiction with the assumption that no integer $\alpha_i$ exists such that $s_{i,i}~\cdot~\alpha_i~\equiv_q~s_{i,m}$.
\item Suppose $p^\ell \nmid s_{m,i}$. Define $x := (0,\ldots,0,p^{k-\ell},0,\ldots,0)^{\text{T}}$ as the vector with $p^{k-\ell}$ in position $i$. Then
    \[Sx \equiv_q (0,\ldots,0,p^{k-\ell}\cdot s_{i,i},0,\ldots,0,p^{k-\ell}\cdot s_{m,i})^{\text{T}}.\]
    Since $p^\ell \mid s_{i,i}$ it follows that $p^{k-\ell}\cdot s_{i,i} \equiv_q 0$. Furthermore, since $p^\ell \nmid s_{m,i}$, it follows that $p^{k-\ell}\cdot s_{m,i} \not\equiv_q 0$, and thereby the system $Sx \equiv_q b$ has a solution for $b := (0,\ldots,0,p^{k-\ell}s_{m,i})$.\qedhere
\end{itemize}
\end{proof}

The contrapositive of Lemma~\ref{lem:no-solution-implies-span} states that if~$a_m \notin \text{span}_q(\{a_1,\ldots,a_{m-1}\})$, then there exists~$c \not \equiv_q 0$ for which~$Ax \equiv_q b$ has a solution. A method to construct such a solution~$x$ follows from our proof above. In the context of capturing a Boolean relation~$R$ by degree-1 polynomials, this constructive proof effectively shows the following: given a prime power~$q$ over which a certain tuple~$u \notin R$ can be captured, one can constructively find the coefficients~$x$ of a polynomial that captures~$u$ by following the steps in the proof.

\newcommand{\rplus}{+}
\newcommand{\rmin}{-}
\newcommand{\rdot}{\cdot}
\begin{proof}[{Proof of Theorem \ref{thm:not-balanced-no-polynomial}}]
Suppose $R$ is not balanced. By Proposition~\ref{prop:balanced-acts-as-alternating}, this implies~$R$ is violated by an alternating operation. Let $f$ be an alternating operation that does not preserve $R$, such that $f(y_1,\ldots,y_m):= \sum_{i=1}^m (-1)^{i+1} y_i$ for some odd~$m$, and for some (not necessarily distinct) $r_1,\ldots,r_m \in R$ we have $f(r_1,\ldots,r_m) = u$ with $u \notin R$.

Suppose for contradiction that there exists a linear polynomial $p_u$ over a ring $E_u$, such that $p_u$ captures $u$ over $E_u$. Let $r_i := (r_{i,1},\ldots,r_{i,k})$ for $i \in [m]$. Since $f(r_1,\ldots,r_m) = u$, we have the following equality over $\mathbb{Z}$:
\begin{equation}\label{eq:balanced-over-ring} u_i = r_{1,i} \rmin r_{2,i} \ldots \rplus r_{m,i}.\end{equation}

Since $r_{j,i} \in \{0,1\}$ for all $i \in [k]$ and $j \in [m]$, equation \eqref{eq:balanced-over-ring} holds over any ring, so in particular over $E_u$.

Let $p_u(x_1,\ldots,x_k)$ be given by $p_u(x_1,\ldots,x_k) := \beta_0 \rplus \beta_1 \rdot x_1 \rplus \beta_2 \rdot x_2 \rplus \dots \rplus \beta_k \rdot x_k$ for ring elements $\beta_0,\ldots,\beta_k$ from $E_u$. By Definition~\ref{def:capture}, $p_u(r_{i,1},\ldots,r_{i,k}) = 0$ for all $i \in [m]$. Thereby the following equalities hold over $E_u$:
\begin{align}
p_u(u_1,\ldots,u_k) &= \beta_0 \rplus \sum_{i=1}^k \beta_i \rdot u_i \nonumber\\
&= \beta_0 \rplus \sum_{i=1}^k \beta_i \rdot (r_{1,i} \rmin r_{2,i}\ldots \rplus r_{m,i}) \nonumber\\
&= \beta_0 \rplus \sum_{i=1}^k \beta_i \rdot r_{1,i} \rmin \beta_i\rdot r_{2,i} \ldots \rplus \beta_i \rdot r_{m,i} \nonumber\\
&= (\beta_0 \rplus \sum_{i=1}^k \beta_i \rdot r_{1,i}) \rmin (\beta_0 \rplus \sum_{i=1}^k \beta_i \rdot r_{2,i}) \dots \rplus (\beta_0 \rplus \sum_{i=1}^k \beta_i \rdot r_{m,i})\label{eq:cancel-beta0}\\
&= p_u(r_{1,1},\ldots,r_{1,k}) \rmin p_u(r_{2,1},\ldots,r_{2,k}) \dots \rplus p_u(r_{m,1},\ldots,r_{m,k})\nonumber\\
&= 0 \nonumber,
\end{align}
where the fourth equality follows from the fact that in line \eqref{eq:cancel-beta0} all but one of the terms $\beta_0$ cancel, since the summation alternates between addition and subtraction. This contradicts the fact that $p_u(u_1,\ldots,u_k)\neq 0$. Thereby, there exists no linear polynomial that captures $u$ with respect to $R$.
\end{proof}
\subsection{Proofs omitted from Section \ref{sec:symmetric}}
\label{app:symmetric}
\begin{proof}[{Proof of Lemma \ref{lem:vector-to-int}}]
We first show the result when $b \leq a$, $b \leq c$, and $b \leq d$. In this case, we use the following tuple to express $x_1\vee x_2$.
\[(\underbrace{\neg x_1,\ldots,\neg x_1}_{(a-b)\text{ copies}}, \underbrace{\neg x_2,\ldots,\neg x_2}_{(c-b)\text{ copies}},\underbrace{1,\ldots,1}_{b\text{ copies}},\underbrace{0,\ldots,0}_{(k-d)\text{ copies}}).\]
Let $f\colon \{x_1,x_2\}\rightarrow \{0,1\}$, then $(\neg f(x_1),\ldots,\neg f(x_1),\neg f(x_2),\ldots,\neg f(x_2),1,\ldots,1,0,\ldots,0)$ has weight $(a-b)(1-f(x_1)) + (c-b)(1-f(x_2)) + b$. It is easy to verify that for $f(x_1) = f(x_2) = 0$, this implies the tuple has weight $a+c-b=d \notin S$ and thus the tuple is not in $R$. Otherwise, the weight is either $a$, $b$, or $c$. In these cases the tuple is contained in $R$, as the weight is contained in $S$.

Note that the above case applies when $b$ is the smallest of all four integers. We now consider the remaining cases. Suppose
$a \leq b$, $a \leq c$, and $a \leq d$ (the case where $c$ is smallest is symmetric by swapping~$a$ and~$c$). In this case, use the tuple
\[(\underbrace{\neg x_1,\ldots,\neg x_1}_{(d-a)\text{ copies}}, \underbrace{ x_2,\ldots, x_2}_{(b-a)\text{ copies}},\underbrace{1,\ldots,1}_{a\text{ copies}},\underbrace{0,\ldots,0}_{(k-c)\text{ copies}}).\]
Consider an assignment $f$ satisfying $x_1 \vee x_2$, verify that the weight of the above tuple under this assignment lies in $\{a,b,c\}$, and thus the tuple is contained in $R$. Assigning $0$ to both $x_1$ and $x_2$ gives weight $d$, such that the tuple is not in $R$.

Otherwise, we have $d \leq a$, $d\leq b$, and $d \leq c$ and use the tuple
\[(\underbrace{ x_1,\ldots, x_1}_{(a-d)\text{ copies}}, \underbrace{ x_2,\ldots, x_2}_{(c-d)\text{ copies}},\underbrace{1,\ldots,1}_{d\text{ copies}},\underbrace{0,\ldots,0}_{(k-b)\text{ copies}}).\]
It is again easy to verify that any assignment to $x_1$ and $x_2$ satisfies this tuple if and only if it satisfies $(x_1 \vee x_2)$, using the fact that $a-d+c = b \in S$.
\end{proof}

\begin{claimproof}[{Proof of Claim \ref{claim:reduce-to-three}}]
If there exist distinct $i,j,\ell\in [m]$ with $i,j$ odd and $\ell$ even, such that $s_i \geq s_\ell \geq s_j$, then these $i,j,\ell$ satisfy the claim statement. Suppose these do not exist, we consider two options.
\begin{itemize}
\item Suppose $s_i \geq s_\ell$ for all $i,\ell\in[m]$ with $i$ odd and $\ell$ even. It is easy to see that thereby, for any $i,j,\ell$ with $i,j$ odd and $\ell$ even it holds that $s_i - s_\ell + s_j \geq 0$. Furthermore, $s_i - s_\ell + s_j\leq s_1 - s_2 +  s_3-s_4 \cdots + s_m = t$ and $t \leq k$ since $t \in U$. Thus, any distinct $i,j,\ell \in [m]$ with $i,j$ odd and $\ell$ even satisfy the statement.
\item Otherwise, $s_i \leq s_\ell$ for all $i,\ell\in[m]$ with $i$ odd and $\ell$ even. It follows that  for any $i,j,\ell$ with $i,j$ odd and $\ell$ even $s_i - s_\ell + s_j \leq k$, as $s_i - s_\ell \leq 0$ and $s_j \leq k$. Furthermore, $s_i - s_\ell + s_j \geq s_1 - s_2 +  s_3-s_4 \cdots + s_m = t$ and $t\geq 0$ by definition. Thus, any distinct $i,j,\ell \in [m]$ with $i,j$ odd and $\ell$ even satisfy the statement.\qedhere
\end{itemize}
\end{claimproof}

\subsection{Proofs omitted from Section \ref{sec:arity-3}}
\label{app:arity-3}
To state the proofs that were omitted from this section, we will first give a number of relevant observations and propositions.
\begin{observation}
\label{obs:defining-lowers-arity}
If a  relation $T$ of arity $m$ is
cone-definable from a relation $U$ of arity $n$,
then $m \leq n$.
\end{observation}

\begin{observation}
\label{obs:defining-still-preserves}
Suppose a relation $T$ is cone-definable from a relation $U$,
and that $g$ is a partial operation that is idempotent
and self-dual.
If $g$ preserves $U$, then $g$ preserves $T$.
\end{observation}

\begin{observation} (transitivity of cone-definability)
Suppose that $T_1$, $T_2$, $T_3$ are relations such that
$T_2$ is cone-definable from $T_1$,
and $T_3$ is cone-definable from $T_2$.
Then $T_3$ is cone-definable from $T_1$.
\end{observation}

\begin{definition}
Let us say that two Boolean relations $T$, $U$
are \emph{cone-interdefinable}
if each is cone-definable from the other.
\end{definition}

The following two propositions are consequences
of Observations~\ref{obs:defining-lowers-arity}
and~\ref{obs:defining-still-preserves}.
We will tacitly use them in the sequel.
Morally, they show that the properties of relations
that we are interested in
are invariant under cone-interdefinability.

\begin{proposition}
If relations $T$, $U$ are cone-interdefinable,
then they have the same arity.
\end{proposition}

\begin{proposition}
\label{prop:interdefinable-preservation}
Suppose that $T$ and $U$ are relations that are cone-interdefinable,
and that $g$ is a partial operation that is idempotent
and self-dual.
Then, $g$ preserves $T$ if and only if $g$ preserves $U$.
\end{proposition}

\begin{proof}[{Proof of Theorem \ref{thm:Boolean-arity-two}}]
Let~$R \subseteq \{0,1\}^2$ be a relation. We prove the two directions separately.

($\Rightarrow$) Proof by contraposition. Suppose that~$R$ is cone-interdefinable with $\kor{2}$. Then in particular,~$R$ cone-defines the $\kor{2}$ relation. Let~$(y_1, y_2)$ be a tuple witnessing cone-definability as in Definition~\ref{def:cone}. Since~$\kor{2}$ is symmetric in its two arguments, we may assume without loss of generality that~$y_i$ is either~$x_i$ or~$\neg x_i$ for~$i \in [2]$. Define~$g \colon \{0,1\}^2 \to \{0,1\}^2$ by letting~$g(i_1,i_2) := (\hat{i}_1, \hat{i}_2)$ where~$\hat{i}_\ell = i_\ell$ if~$y_i = x_i$ and~$\hat{i}_\ell = 1 - i_{\ell}$ if~$y_i = \neg x_i$. By definition of cone-definability we then have~$g(1,0), g(0,1), g(1,1) \in R$ while~$g(0,0) \notin R$. But~$g(1,0) - g(1,1) + g(0,1) = g(0,0)$, showing that~$R$ is not preserved by all alternating operations and therefore is not balanced.

($\Leftarrow$) We again use contraposition. Suppose~$R$ is not balanced; we will prove~$R$ is cone-interdefinable with $\kor{2}$. Let~$f \colon \{0,1\}^k \to \{0,1\}$ be a balanced partial Boolean operation \emph{of minimum arity} that does not preserve~$R$. Let~$\alpha_1, \ldots, \alpha_k \in \mathbb{Z}$ be the coefficients of~$f$, as in Definition~\ref{def:balanced_op}. Let~$s^1, \ldots, s^k \in R$ such that~$f(s^1, \ldots, s^k) = u \in \{0,1\}^2 \setminus R$ witnesses that~$f$ does not preserve~$R$. By Definition~\ref{def:balanced_op} we have~$u = \sum_{i=1}^k \alpha_i s^i$ and~$\sum_{i=1}^k \alpha_i = 1$. This shows that if~$\alpha_i = 0$ for some coordinate~$i$, then that position does not influence the value of~$f$, implying the existence of a smaller-arity balanced relation that does not preserve~$T$. Hence our choice of~$f$ as a minimum-arity operation ensures that~$\alpha_i$ is nonzero for all~$i \in [k]$.

\begin{claim}
The tuples~$s^1, \ldots, s^k$ are all distinct.
\end{claim}
\begin{claimproof}
Suppose that~$s^i = s^j$ for some distinct~$i,j \in [k]$, and assume without loss of generality that~$i = k-1$ and~$j=k$. But then the balanced operation~$f'$ of arity~$k-1$ defined by the coefficients~$(\alpha_1, \ldots, \alpha_{k-2}, \alpha_{k-1} + \alpha_k)$ does not preserve~$T$ since~$f'(s^1, \ldots, s^{k-1}) = f(s^1, \ldots, s^k) = u \notin R$, contradicting that~$f$ is a minimum-arity balanced operation that does not preserve~$R$.
\end{claimproof}

\begin{claim}
The arity~$k$ of operation~$f$ is~$3$.
\end{claim}
\begin{claimproof}
Since~$s^1, \ldots, s^k \in R \subseteq \{0,1\}^2$ are all distinct, while~$u \in \{0,1\}^2 \setminus R$, we have~$k \leq 3$. We cannot have~$k=1$ since that would imply~$f(s^1) = s^1 \in R$ and~$f(s^1, \ldots, s^k) = f(s^1) = u \notin R$. It remains to show that~$k \neq 2$. So assume for a contradiction that~$k=2$. Since~$s^1$ and~$s^2$ are distinct, there is a position~$\ell \in [2]$ such that~$s^1_\ell \neq s^2_\ell$. Assume without loss of generality that~$s^1_\ell = 1$ while~$s^2_\ell = 0$. Since~$f(s^1, \ldots, s^k) = f(s^1, s^2) = \alpha_1 s^1 + \alpha_2 s^2 = u \in \{0,1\}^2 \setminus R$, we find~$u_\ell = \alpha_1 s^1_\ell + \alpha_2 s^2_\ell = \alpha_1 \cdot 1 + \alpha_2 \cdot 0 \in \{0,1\}$. Since~$\alpha_1$ is a nonzero integer, we must have~$\alpha_1 = 1$. But since~$\alpha_1 + \alpha_2 = 1$ by definition of a balanced operation, this implies~$\alpha_2 = 0$, contradicting that~$f$ is a minimum-arity balanced operation that does not preserve~$R$.
\end{claimproof}

The previous two claims show that there are at least three distinct tuples in~$R \subseteq \{0,1\}^2$. Since~$u \in \{0,1\}^2 \setminus R$ it follows that~$|R| = 3$. Hence~$R$ and \kor{2} are both Boolean relations of arity two that each have three tuples. To cone-define one from the other, one may easily verify that it suffices to use the tuple~$(y_1, y_2)$, where~$y_i = x_i$ if~$u_i = 0$ and~$y_i = \neg x_i$ otherwise.
\end{proof}

In order to prove Theorem \ref{thm:not-balanced-implies-or}, we first present some additional lemmas and definitions. 
Let $U \subseteq \{ 0, 1 \}^n$ be a relation.
We say that $w \in \{ 0, 1 \}^n$
is a \emph{witness} for $U$ if
$w \notin U$, and there exists a balanced operation
$f \colon \{ 0, 1 \}^k \to \{ 0, 1 \}$
and tuples $t^1, \ldots, t^k \in U$
such that $w = f(t^1, \ldots, t^k)$.
Observe that $U$ is not balanced
if and only if there exists a witness for $U$.

\begin{lemma}
\label{lemma:witness-modulo}
Suppose that $U \subseteq \{ 0, 1 \}^n$ is a Boolean relation,
and
that there exist an integer $c$ and a natural number
$m > 1$ such that, for each $u \in U$,
it holds that
$$\weight(u) \equiv_m c.$$
Then, if $w$ is a witness for $U$,
it holds that $\weight(w) \equiv_m c$.
\end{lemma}
\begin{proof}
Since $w$ is a witness for $U$,
there exist tuples
$t^1 = (t^1_1, \ldots, t^1_n)$, $\ldots$,
$t^k = (t^k_1, \ldots, t^k_n)$
and a balanced operation $f \colon \{ 0,1 \}^k \to \{ 0, 1 \}$
such that
$f(t^1, \ldots, t^k) = w$.
Let $\alpha_1, \ldots, \alpha_k$ be the coefficients of $f$.
From
$f(t^1, \ldots, t^k) = w$, we obtain that
$ \alpha_1 \weight(t^1) +
  \cdots
+ \alpha_k \weight(t^k) = \weight(w) $.
Since $\sum _{i\in[k]}\alpha_i = 1$ by definition of a balanced operation, we have
$$\alpha_1 \weight(t^1) +
  \cdots
+ \alpha_k \weight(t^k) \equiv_m
\alpha_1 c + \cdots + \alpha_k c = \left(\sum_{i \in [k]} \alpha_i \right)c = c$$
and the result follows.
\end{proof}

We will view (Boolean) tuples of arity $n$
as maps $f \colon [n] \to \{ 0, 1 \}$,
via the natural correspondence where
such a map $f$ represents the tuple
$(f(1),\ldots,f(n))$.  We freely interchange between
these two representations of tuples.

For~$S \subseteq \mathbb{N}$, we say that $f \colon S \to \{ 0, 1 \}$
is a \emph{no-good} of $U \subseteq \{0,1\}^n$
when:
\begin{itemize}
	\item $S \subseteq [n]$;
	\item each extension $g \colon [n] \to \{ 0, 1 \}$ of $f$ is not an element of $U$; and
	\item there exists an extension $h \colon [n] \to \{ 0, 1 \}$ of $f$ that is a witness for $U$.
\end{itemize}
We say that $f \colon S \to \{ 0, 1 \}$
is a \emph{min-no-good} if $f$ is a no-good,
but no proper restriction of $f$ is a no-good.
Observe that the following are equivalent,
for a relation:
the relation is not balanced; it has a witness;
it has a no-good; it has a min-no-good.

When $U \subseteq \{ 0, 1 \}^n$ is a relation and $S \subseteq [n]$,
let $s_1 < \cdots < s_m$ denote the elements of $S$; then, we use
$U \restr S$ to denote the relation $\{ (h(s_1), \ldots, h(s_m)) ~|~ h \in U \}$.

\begin{proposition}
\label{prop:project-min-no-good}
Let $U \subseteq \{ 0, 1 \}^n$ be a relation,
let $S \subseteq [n]$,
and suppose that $f \colon S \to \{ 0, 1 \}$ is a min-no-good of $U$.
Then $f$ is a min-no-good of $U \restr S$.
\end{proposition}

\begin{proof}
Observe that $f$ is not in $U \restr S$;
since $f$ has an extension that is a witness for $U$,
it follows that $f$ is a witness for $U \restr S$.
Thus, $f$ is a no-good of $U \restr S$.
In order to obtain that $f$ is a min-no-good of $U \restr S$,
it suffices to establish that,
for any restriction $f^- \colon S^- \to \{ 0, 1 \}$ of $f$,
it holds that $f^-$ is a no-good of $U$
if and only if $f^-$ is a no-good of $U \restr S$.
This follows from what we have established concerning $f$
and the following fact: all extensions
$h \colon S \to \{ 0, 1 \}$ of $f^-$ are not in $U \restr S$
if and only if
all extensions
$h' \colon [n] \to \{ 0, 1 \}$ of $f^-$ are not in $U$.
\end{proof}

Using these tools we are finally in position to prove Theorem~\ref{thm:not-balanced-implies-or}.

\begin{proof}[{Proof of Theorem \ref{thm:not-balanced-implies-or}}]
Let $f \colon S \to \{ 0, 1 \}$ be a min-no-good of $U$.

It cannot hold that $|S| = 0$, since then $U$
would be empty and hence preserved by all balanced operations.
It also cannot hold that $|S| = 1$, since then
$f$ would be a min-no-good of $U \restr S$
(by Proposition~\ref{prop:project-min-no-good}),
which is not possible since $U \restr S$ would have arity $1$
and hence would be preserved by all balanced operations
(by Observation~\ref{obs:Boolean-arity-one}).

For the remaining cases, by replacing $U$ with a
relation that is interdefinable with it,
we may assume that $f \colon S \to \{ 0, 1 \}$
maps each $s \in S$ to $0$.

Suppose that $|S| = 2$, and assume for the sake
of notation that $S = \{ 1, 2 \}$
(this can be obtained by replacing $U$
with a relation that is interdefinable with it).
By
Proposition~\ref{prop:project-min-no-good},
$f$ is a min-no-good of $U \restr S$.
By Theorem~\ref{thm:Boolean-arity-two},
we obtain that $U \restr S$
contains all tuples other than $f$,
that is,
we have $\{ (0,1),(1,0),(1,1) \} = U \restr S$.
It follows that there exists a \emph{realization},
where we define a \emph{realization} to be a tuple
$(a_1,a_2,a_3) \in \{0,1\}^3$
such that
$(0,1,a_1),(1,0,a_2),(1,1,a_3) \in U$.
Let us refer to $(0,0,1)$ and $(1,1,0)$ as \emph{bad} tuples,
and to all other arity $3$ tuples as \emph{good} tuples.

\begin{claim}
If there is a realization that is a good tuple, then the \kor{2} relation is cone-definable from~$U$.
\end{claim}
\begin{claimproof}
We show cone-definability via a tuple of the form $(x_1,x_2,y)$ where $y \in \{ 0,1,x_1,x_2, \linebreak[0] \neg x_1,\neg x_2 \}$. The right setting for~$y$ can be derived from the realization that forms a good tuple.
\begin{itemize}
	\item choose $y = 0$ for $(0,0,0)$;
	\item $y = 1$ for $(1,1,1)$;
	\item $y=x_1$ for $(0,1,1)$;
	\item $y=x_2$ for $(1,0,1)$;
	\item $y=\neg x_1$ for $(1,0,0)$; and,
	\item $y=\neg x_2$ for $(0,1,0)$.
\end{itemize}
It is easy to verify that this choice of~$y$ gives the desired cone-definition.
\end{claimproof}

\begin{claim}
There is a realization that is a good tuple.
\end{claim}
\begin{claimproof}
Proof by contradiction. If there exists no realization that is a good tuple,
every realization is a bad tuple; moreover,
there is a unique realization, for if there were more than one,
there would exist a realization that was a good tuple.
We may assume (up to interdefinability of $U$)
that the unique realization is $(1,1,0)$.
Then, $U$ is the relation
$\{ (0,1,1),(1,0,1),(1,1,0) \}$
containing exactly the weight $2$ tuples; applying
Lemma~\ref{lemma:witness-modulo}
to $U$ with $a = 2$ and $m = 3$,
we obtain that for any witness $w$ for $U$,
it holds that $\weight(w) \equiv_3 2$.
This implies that $f$ has no extension $w'$ that is a witness,
since any such extension must have
$\weight(w')$ equal to $0$ or $1$ as~$f$ maps both~$s \in S$ to~$0$;
we have thus contradicted
that $f$ is a no-good of $U$.
\end{claimproof}
Together, the two claims complete the case that~$|S|=2$.

Suppose that $|S| = 3$.
Since $f$ is both a min-no-good and a witness, mapping all~$s \in S$ to~$0$,
it follows that each of the weight $1$
tuples $(1,0,0),(0,1,0),(0,0,1)$ is contained in $U$.
We claim that $U$ contains a weight $2$ tuple;
if not, then $U$ would contain only weight $1$ and weight $3$
tuples, and by invoking Lemma~\ref{lemma:witness-modulo}
with $a = 1$ and $m = 2$, we would obtain
that $\weight(f) \equiv_2 1$, a contradiction.
Assume for the sake of notation that $U$
contains the weight $2$ tuple $(0,1,1)$.
Then $U$ cone-defines the $\kor{2}$ relation via
the tuple $(0,x_1,x_2)$, since $(0,0,0) \notin R$
and $(0,1,0),(0,0,1),(0,1,1) \in R$.
\end{proof}

\end{document}